\documentclass[]{article}
\usepackage[nolinenums]{custom_arxiv}

\usepackage[utf8]{inputenc} 
\usepackage[T1]{fontenc}    
\usepackage{hyperref}       
\usepackage{url}            
\usepackage{booktabs}       
\usepackage{amsfonts}       
\usepackage{nicefrac}       
\usepackage{microtype}      
\usepackage{amsmath}
\usepackage{amsthm}
\usepackage{xcolor}
\usepackage{graphicx}
\usepackage{amssymb}
\usepackage{wrapfig}
\usepackage{etoolbox}
\usepackage{thm-restate}
\usepackage{mathtools}
\usepackage{mleftright}
\usepackage{enumitem}
\usepackage{multirow}
\usepackage[noend]{algpseudocode}
\usepackage{tikz}
\usetikzlibrary{arrows.meta,calc}
\usepackage[capitalise,noabbrev]{cleveref}
\usepackage{bm}

\title{Efficient Regret Minimization Algorithm for Extensive-Form Correlated Equilibrium\thanks{This paper was accepted for publication at NeurIPS 2019.}}
\author{Gabriele Farina\\
Computer Science Department\\
Carnegie Mellon University\\
\texttt{gfarina@cs.cmu.edu}
\And Chun Kai Ling\\
Computer Science Department\\
Carnegie Mellon University\\
\texttt{chunkail@cs.cmu.edu}
\And Fei Fang\\
Institute for Software Research\\
Carnegie Mellon University\\
\texttt{feif@cs.cmu.edu}
\And Tuomas Sandholm\\
Computer Science Department, CMU\\
Strategic Machine, Inc.\\
Strategy Robot, Inc.\\
Optimized Markets, Inc.\\
\texttt{sandholm@cs.cmu.edu}
}

\ifcsmacro{assumption}{}{
	
}
\ifcsmacro{corollary}{}{
	
}
\ifcsmacro{definition}{}{
	\newtheorem{definition}{Definition}[]
}
\ifcsmacro{example}{}{
	
}
\ifcsmacro{fact}{}{
	
}
\ifcsmacro{informal_theorem}{}{
	
}
\ifcsmacro{lemma}{}{
	\newtheorem{lemma}{Lemma}[]
}
\ifcsmacro{proposition}{}{
	\newtheorem{proposition}{Proposition}[]
}
\ifcsmacro{remark}{}{
	\newtheorem{remark}{Remark}[]
}
\ifcsmacro{theorem}{}{
	\newtheorem{theorem}{Theorem}[]
}
\ifcsmacro{property}{}{
	
}
\DeclareMathOperator*{\ext}{\,\triangleleft\,}

\newcommand{\cX}{\mathcal{X}}
\newcommand{\cY}{\mathcal{Y}}
\newcommand{\cZ}{\mathcal{Z}}

\newcommand{\bbR}{\mathbb{R}}
\mathtoolsset{centercolon}
\newcommand{\defeq}{\mathrel{:\mkern-0.25mu=}}

\newcommand{\symp}[1]{\Delta^{\!#1}}
\renewcommand{\vec}[1]{\bm{#1}}
\newcommand{\mat}[1]{\bm{#1}}
\newcommand{\regm}[1]{RM$_{\mathcal{#1}}$}
\newcommand{\rele}{\mathrel{\triangleright\mkern-1.7mu\triangleleft}}
\newcommand{\emptyseq}{\varnothing}
\newcommand*\circled[1]{\tikz[baseline=(char.base)]{
            \node[shape=circle,draw,inner sep=.5pt] (char) {\fontsize{8pt}{9pt}\selectfont{#1}};}}
\newcommand{\conn}{\rightleftharpoons}
\newcommand{\xileaf}[2]{\xi_i(#1; #2)}
\newcommand{\seqf}[1]{Q_{#1}}
\DeclareMathOperator*{\argmin}{argmin}

\makeatletter
\g@addto@macro \normalsize {%
 \addtolength\abovedisplayskip{-2pt}%
 \addtolength\belowdisplayskip{-2pt}%
}
\makeatother

\begin{document}
    \maketitle
    \begin{abstract}
      Self-play methods based on regret
      minimization have become the state of the art for computing
      Nash equilibria in large two-players zero-sum
      extensive-form games. These methods fundamentally rely on the
      hierarchical structure of the players' sequential strategy spaces to
      construct a regret minimizer that recursively minimizes regret at
      each decision point in the game tree. In this paper, we introduce the
      first efficient regret minimization algorithm for computing 
      extensive-form
      \emph{correlated} equilibria in large two-player \emph{general-sum} 
      games with no
      chance moves. Designing such an algorithm is significantly more 
      challenging than designing one for the Nash
      equilibrium counterpart, as the constraints that define the
      space of correlation plans lack the hierarchical structure and might even form cycles. We show that some of the constraints are redundant and can be excluded from consideration, and present an efficient algorithm that generates the space of extensive-form correlation plans incrementally from the remaining constraints. This structural decomposition is achieved via a special
      convexity-preserving operation that we coin \emph{scaled extension}.
      We show that a regret minimizer can be designed for a scaled extension of any two convex sets, and that from the decomposition we then obtain a global regret minimizer.
      Our algorithm produces feasible iterates. Experiments show that it significantly outperforms prior approaches and for larger problems it is the only viable option.
    \end{abstract}

    \section{Introduction}

In recent years, self-play methods based on regret minimization, such as counterfactual regret minimization (CFR)~\citep{Zinkevich07:Regret} and its faster variants~\citep{Tammelin15:Solving,Brown17:Dynamic,Brown19:Solving} have emerged as powerful tools for computing Nash equilibria in large extensive-form games, and have been instrumental in several recent milestones in poker~\citep{Bowling15:Heads,Brown17:Safe,Brown17:Superhuman,Moravvcik17:DeepStack,Brown19:Superhuman}. These methods exploit the hierarchical structure of the sequential strategy spaces of the players to construct a regret minimizer that recursively minimizes regret {locally} at each decision point in the game tree. This has inspired regret-based algorithms for other solution concepts in game theory, such as  extensive-form perfect equilibria~\citep{Farina17:Regret}, Nash equilibrium with strategy constraints~\citep{Farina17:Regret,Farina19:Online,Farina19:Regret,Davis19:Solving}, and quantal-response equilibrium~\citep{Farina19:Online}.

In this paper, we give the first {efficient} regret-based algorithm for finding an \textit{extensive-form correlated equilibrium (EFCE)}~\citep{Stengel08:Extensive} in two-player general-sum games with no chance moves.
EFCE is a natural extension of the correlated equilibrium (CE) solution concept to the setting of extensive-form games. Here, the strategic interaction of rational players is complemented by a \emph{mediator} that privately recommends behavior, but does not \emph{enforce it}: it is up to the mediator to make recommendations that the players are incentivized to follow.
Designing a regret minimization algorithm that can efficiently search over the space of extensive-form correlated strategies (known as \emph{correlation plans}) is significantly more difficult than designing one for Nash equilibrium. This is because the constraints that define the space of correlation plans lack the hierarchical structure of sequential strategy spaces and might even form cycles. Existing general-purpose regret minimization algorithms, such as follow-the-regularized-leader~\citep{Shalev07:Primal} and mirror descent, as well as those proposed by~\citet{Gordon08:No} in the context of convex games, are not practical: they require the evaluation of proximal operators (generalized projections problems) or the minimization of linear functions on the space of extensive-form correlation plans. In the former case, no distance-generating function is known that can be minimized efficiently over this space, while in the latter case current linear programming technology does not scale to large games, as we show in the experimental section of this paper. The regret minimization algorithm we present in this paper computes the next iterate in \emph{linear} time in the dimension of the space of correlation plans.

We show that some of the constraints that define the polytope of correlation plans are redundant and can be eliminated, and present an efficient algorithm that generates the space of correlation plans incrementally from the remaining constraints. This structural decomposition is achieved via a special convexity-preserving operation that we coin \emph{scaled extension}.
We show that a regret minimizer can be designed for a scaled extension of any two convex sets, and that from the decomposition we then obtain a global regret minimizer.
Experiments show that our algorithm significantly outperforms prior approaches---the LP-based approach~\citep{Stengel08:Extensive} and a very recent subgradient descent algorithm~\citep{Farina19:Correlation}---and for larger problems it is the only viable option.

    \vspace{-1mm}
\section{Preliminaries}\label{sec:efg}
\vspace{-2mm}

Extensive-form games (EFGs) are played on a game tree. 
Each node in the game tree belongs to a player, who acts at that node; for the purpose of this paper, we focus on two-player games only. Edges leaving a node correspond to actions that can be taken at that node. In order to capture private information, the game tree is supplemented with \emph{information sets}. Each node belongs to exactly one information set, and each information set is a nonempty set of tree nodes for the same Player $i$, which are the set of nodes that Player $i$ cannot distinguish among, given what they have observed so far. We will focus on \emph{perfect-recall} EFGs, that is, EFGs where no player forgets what the player knew earlier. We denote by $\mathcal{I}_1$ and $\mathcal{I}_2$ the sets of all information sets that belong to Player 1 and 2, respectively. All nodes that belong to an information set $I \in \mathcal{I}_1 \cup \mathcal{I}_2$ share the same set of available actions (otherwise the player acting at those nodes would be able to distinguish among them);
we denote by $A_I$ the set of actions available at information set $I$. We define the set of \emph{sequences} of Player $i$ as the set $\Sigma_i \defeq \{(I, a) : I \in \mathcal{I}_i, a \in A_I\} \cup \{\emptyseq\}$,
where the special element $\emptyseq$ is called \emph{empty sequence}. Given an information set $I \in \mathcal{I}_i$, we denote by $\sigma(I)$ the \emph{parent sequence} of $I$, defined as the last pair $(I', a') \in \Sigma_i$
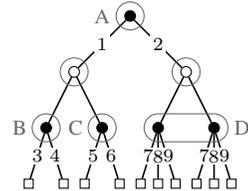
\begin{wrapfigure}{r}{3.4cm}\centering
\vspace{-2mm}
\def\done{.8}
\def\dtwo{.40}
\def\dleaf{.25}
\hfill\begin{tikzpicture}[baseline=0pt,scale=.93]
  \node[fill=black,draw=black,circle,inner sep=.5mm] (A) at (0, 0) {};
  \node[fill=white,draw=black,circle,inner sep=.5mm] (X) at ($(-\done,-.8)$) {};
  \node[fill=white,draw=black,circle,inner sep=.5mm] (Y) at ($(\done,-.8)$) {};
  \node[fill=black,draw=black,circle,inner sep=.5mm] (B) at ($(X) + (-\dtwo, -.8)$) {};
  \node[fill=black,draw=black,circle,inner sep=.5mm] (C) at ($(X) + (\dtwo, -.8)$) {};
  \node[fill=white,draw=black,inner sep=.6mm] (l1) at ($(B) + (-\dleaf, -.8)$) {};
  \node[fill=white,draw=black,inner sep=.6mm] (l2) at ($(B) + (\dleaf, -.8)$) {};
  \node[fill=white,draw=black,inner sep=.6mm] (l3) at ($(C) + (-\dleaf, -.8)$) {};
  \node[fill=white,draw=black,inner sep=.6mm] (l4) at ($(C) + (\dleaf, -.8)$) {};

  \node[fill=black,draw=black,circle,inner sep=.5mm] (D1) at ($(Y) + (-\dtwo, -.8)$) {};
  \node[fill=black,draw=black,circle,inner sep=.5mm] (D2) at ($(Y) + (\dtwo, -.8)$) {};
  \node[fill=white,draw=black,inner sep=.6mm] (l5) at ($(D1) + (-\dleaf, -.8)$) {};
  \node[fill=white,draw=black,inner sep=.6mm] (l6) at ($(D1) + (0, -.8)$) {};
  \node[fill=white,draw=black,inner sep=.6mm] (l7) at ($(D1) + (\dleaf, -.8)$) {};
  \node[fill=white,draw=black,inner sep=.6mm] (l8) at ($(D2) + (-\dleaf, -.8)$) {};
  \node[fill=white,draw=black,inner sep=.6mm] (l9) at ($(D2) + (0, -.8)$) {};
  \node[fill=white,draw=black,inner sep=.6mm] (l10) at ($(D2) + (\dleaf, -.8)$) {};

  \draw[semithick] (A) --node[fill=white,inner sep=.9] {\scriptsize$1$} (X);
  \draw[semithick] (A) --node[fill=white,inner sep=.9] {\scriptsize$2$} (Y);
  \draw[semithick] (B) --node[fill=white,inner sep=.9] {\scriptsize$3$} (l1);
  \draw[semithick] (B) --node[fill=white,inner sep=.9] {\scriptsize$4$} (l2);
  \draw[semithick] (C) --node[fill=white,inner sep=.9] {\scriptsize$5$} (l3);
  \draw[semithick] (C) --node[fill=white,inner sep=.9] {\scriptsize$6$} (l4);
  \draw[semithick] (D1) --node[fill=white,inner xsep=0,inner ysep=.9,xshift=-.4] {\scriptsize$7$} (l5);
  \draw[semithick] (D1) --node[fill=white,inner xsep=0,inner ysep=.9] {\scriptsize$8$} (l6);
  \draw[semithick] (D1) --node[fill=white,inner xsep=0,inner ysep=.9,xshift=.6] {\scriptsize$9$} (l7);
  \draw[semithick] (D2) --node[fill=white,inner xsep=0,inner ysep=.9,xshift=-.4] {\scriptsize$7$} (l8);
  \draw[semithick] (D2) --node[fill=white,inner xsep=0,inner ysep=.9] {\scriptsize$8$} (l9);
  \draw[semithick] (D2) --node[fill=white,inner xsep=0,inner ysep=.9,xshift=.6] {\scriptsize$9$} (l10);
  \draw[semithick] (B) -- (X) -- (C);
  \draw[semithick] (D1) -- (Y) -- (D2);

    \draw[black!60!white] (X) circle (.2);
    \draw[black!60!white] (Y) circle (.2);

  \draw[black!60!white] (A) circle (.2);
  \node[black!60!white]  at ($(A) + (-.4, 0)$) {\textsc{a}};

  \draw[black!60!white] (B) circle (.2);
  \node[black!60!white]  at ($(B) + (-.38, 0)$) {\textsc{b}};

  \draw[black!60!white] (C) circle (.2);
  \node[black!60!white]  at ($(C) + (-.38, 0)$) {\textsc{c}};

  \draw[black!60!white] ($(D1) + (0, .2)$) arc (90:270:.2);
  \draw[black!60!white] ($(D1) + (0, .2)$) -- ($(D2) + (0, .2)$);
  \draw[black!60!white] ($(D1) + (0, -.2)$) -- ($(D2) + (0, -.2)$);
  \draw[black!60!white] ($(D2) + (0, -.2)$) arc (-90:90:.2);
  \node[black!60!white]  at ($(D2) + (.4, 0)$) {\textsc{d}};
\end{tikzpicture}\hfill
\vspace{-1mm}
\caption{\small Small example.}
\label{fig:small sf}
\vspace{-5mm}
\end{wrapfigure}
encountered on the path from the root to any node $v \in I$; if no such pair exists (that is, Player $i$ never acts before any node $v\in I$), we let $\sigma(I) = \emptyseq$.
We (recursively) define a sequence $\tau\in\Sigma_i$ to be a \emph{descendent} of sequence $\tau'\in\Sigma_i$, denoted by $\tau \succeq \tau'$, if $\tau = \tau'$ or if $\tau=(I,a)$ and $\sigma(I) \succeq \tau'$. We use the notation $\tau \succ \tau'$ to mean $\tau\succeq\tau' \land \tau \neq \tau'$. \cref{fig:small sf} shows a small example EFG; black round nodes belong to Player 1, white round nodes belong to Player 2, action names are not shown, gray round sets define information sets, and the numbers along the edges define concise names for sequences (for example, `7' denotes sequence $(\textsc{d}, a)$ where $a$ is the leftmost action at \textsc{d}).

\textbf{Sequence-Form Strategies}\quad In the sequence-form representation~\citep{Romanovskii62:Reduction,Koller96:Efficient,Stengel96:Efficient}, a strategy for Player $i$ is compactly represented via a vector $\vec{x}$ indexed by sequences $\sigma \in \Sigma_i$. When $\sigma=(I,a)$, the entry ${x}[\sigma] \ge 0$ defines the product of the probabilities according to which Player $i$ takes their actions on the path from the root to information set $I$, up to and including action $a$; furthermore, $x[\emptyseq]= 1$. Hence, in order to be a valid sequence-form strategy, $x$ must satisfy the `probability mass conservation' constraint: for all $I\in\mathcal{I}_i$, $\sum_{a\in A_I} {x}[(I, a)] = x[\sigma(I)]$. That is, every information sets partitions the probability mass received from the parent sequence onto its actions. In this sense, the constraints that define the space of sequence-form strategies naturally exhibit a hierarchical structure.

\vspace{-1mm}
\subsection{Extensive-Form Correlated Equilibria}\label{sec:efce}
\vspace{-1mm}

\textit{Extensive-form correlated equilibrium (EFCE)}~\citep{Stengel08:Extensive} is a natural extension of the solution concept of \textit{correlated equilibrium (CE)}~\citep{Aumann74:Subjectivity} to extensive-form games. In EFCE, a mediator privately reveals recommendations to the players as the game progresses. These recommendations are \emph{incremental}, in the sense that recommendations for the move to play at each decision point of the game are revealed only if and when the decision point is reached. This is in contrast with CE, where recommendations \emph{for the whole game} are privately revealed upfront when the game starts. Players are free to not follow the recommended moves, but once a player does not follow a recommendation, he will not receive further recommendations. In an EFCE, the recommendations are incentive-compatible---that is, the players are motivated to follow all recommendations.  EFCE and CE are good candidates to model strategic interactions in which intermediate forms of centralized control can be achieved~\citep{Ashlagi08:Value}.

In a recent preprint, \citet{Farina19:Correlation} show that in two-player perfect-recall extensive-form games, an EFCE that guarantees a social welfare (that is, sum of player's utilities) at least $\tau$ is the solution to a bilinear saddle-point problem, that is an optimization problem of the form
$
  \min_{\vec{x}\in\cX}\max_{\vec{y}\in\cY} \vec{x}^{\!\top}\!\! \mat{A} \vec{y},$
where $\cX$ and $\cY$ are convex and compact sets and $\mat{A}$ is a matrix of real numbers. In the case of EFCE, $\cX = \Xi$ is known as the \emph{polytope of correlation plans} (see \cref{sec:xi polytope}) and $\cY$ is the convex hull of certain sequence-form strategy spaces.
In general, $\Xi$ cannot be captured by a polynomially small set of constraints, since computing an optimal EFCE in a two-player perfect-recall game is computationally hard~\citep{Stengel08:Extensive}.\footnote{A feasible EFCE can be found in theoretical polynomial time~\citep{Huang08:Computing,Huang11:Equilibrium} using the \emph{ellipsoid-against-hope} algorithm~\citep{Papadimitriou08:Computing,Jiang15:Polynomial}. Unfortunately, that algorithm is known to not scale beyond small games.} However, in the special case of games with no chance moves, this is not the case, and $\Xi$ is the intersection of a polynomial (in the game tree size) number of constraints, as discussed in the next subsection.
In fact, most of the current paper is devoted to studying the structure of $\Xi$. We will largely ignore $\cY$, for which an efficient regret minimizer can already be built, for instance by using the theory of \emph{regret circuits}~\citep{Farina19:Regret} (see also \cref{app:efce saddle point}). Similarly, we will not use any property of matrix $\mat{A}$ (except that it can be computed and stored efficiently).

\vspace{-1mm}
\subsection{Polytope of Extensive-Form Correlation Plans in Games with no Chance Moves}\label{sec:xi polytope}
\vspace{-1mm}

In their seminal paper, \citet{Stengel08:Extensive} characterize the constraints that define the space of extensive-form correlation plans $\Xi$ in the case of two-player perfect-recall games \emph{with no chance moves}. The characterization makes use of the following two concepts:

\begin{definition}[Connected information sets, $I_1 \conn I_2$]\label{def:connected infosets}
 Let $I_1,I_2$ be information sets for Player 1 and 2, respectively. We say that $I_1$ and $I_2$ are \emph{connected}, denoted $I_1 \conn I_2$, if there exist two nodes $u \in I_1, v \in I_2$ such that $u$ is on the path from the root to $v$, or $v$ is on the path from the root to $u$.
\end{definition}

\begin{definition}[Relevant sequence pair, $\sigma_1 \rele \sigma_2$]\label{def:relevant seq pair}
  Let $\sigma_1 \in \Sigma_1, \sigma_2 \in \Sigma_2$. We say that $(\sigma_1,\sigma_2)$ is a \emph{relevant sequence pair}, and write $\sigma_1 \rele \sigma_2$, if either $\sigma_1$ or $\sigma_2$ or both is the empty sequence, or if $\sigma_1 = (I_1, a_1)$ and $\sigma_2 = (I_2, a_2)$ and $I_1 \conn I_2$. Similarly, given $\sigma_1 \in \Sigma_1$ and $I_2 \in \mathcal{I}_2$, we say that $(\sigma_1, I_2)$ forms a relevant sequence-information set pair, and write $\sigma_1 \rele I_2$, if $\sigma_1 = \emptyseq$ or if $\sigma_1 = (I_1, a_1)$ and $I_1 \conn I_2$ (a symmetric statement holds for $I_1 \rele \sigma_2$).
\end{definition}

\begin{definition}[\citet{Stengel08:Extensive}]\label{def:xi}
  In a two-player perfect-recall extensive-form game with no chance moves, the space $\Xi$ of \emph{correlation plans} is a convex polytope containing \emph{nonnegative} vectors indexed over relevant sequences pairs, and is defined as
\begin{equation*}\small
  \Xi \defeq \left\{\vec{\xi} \ge \vec{0}: \begin{array}{ll}
    \bullet\ \ \xi[\emptyseq, \emptyseq] = 1 \\[.5mm]
    \bullet\ \ \!\sum_{a \in A_I}~\!\xi[(I_1, a),\hspace{.1mm} \sigma_2] = \xi[\sigma(I_1),\hspace{.4mm} \sigma_2] & \forall I_1 \in \mathcal{I}_1, \sigma_2 \in \Sigma_2 \ \hspace{.0mm} \text{ s.t. } I_1 \rele \sigma_2\\[.5mm]
    \bullet\ \ \!\sum_{a \in A_J} \xi[\sigma_1, (I_2, a)] = \xi[\sigma_1, \sigma(I_2)] & \forall I_2 \in \mathcal{I}_2, \sigma_1 \in \Sigma_1 \,\text{ s.t. } \sigma_1 \rele I_2
  \end{array}\!\!\!\right\}\!.
\end{equation*}
In particular, $\Xi$ is the intersection of at most $1+|\mathcal{I}_1|\cdot |\Sigma_2|+|\Sigma_1|\cdot |\mathcal{I}_2|$ constraints, a polynomial number in the input game size.
\end{definition}

\subsection{Regret Minimization and Relationship with Bilinear Saddle-Point Problems}\label{sec:regret}

A regret minimizer is a device that supports two operations: (i) \textsc{Recommend}, which provides the next decision $\vec{x}^{t+1}\in\cX$, where $\cX$ is a nonempty, convex, and compact subset of a Euclidean space $\bbR^n$; and (ii) \textsc{ObserveLoss}, which receives/observes a convex loss function ${\ell}^t$ that is used to evaluate decision $\vec{x}^t$~\citep{Zinkevich03:Online}. In this paper, we will consider linear loss functions, which we represent in the form of a vector $\vec{\ell}^t \in \bbR^n$. A regret minimizer is an \emph{online} decision maker in the sense that each decision is made by taking into account only past decisions and their corresponding losses. The quality metric for the regret minimizer is its \emph{cumulative regret} $R^T$, defined as the difference between the loss cumulated by the sequence of decisions $\vec{x}^1, \dots, \vec{x}^T$ and the loss that would have been cumulated by the \emph{best-in-hindsight time-independent} decision $\hat{\vec{x}}$. Formally,
$
  R^T \defeq \sum_{t=1}^T \langle \vec{\ell}^t, \vec{x}^t\rangle - \min_{\hat{\vec{x}} \in \cX} \sum_{t=1}^T \langle \vec{\ell}^t , \hat{\vec{x}}\rangle.
$
A `good' regret minimizer has $R^T$ sublinear in $T$; this property is known as \emph{Hannan consistency}.
%
%
Hannan consistent regret minimizers can be used to converge to a solution of a \emph{bilinear saddle-point problem} (\cref{sec:efce}). To do so, two regret minimizers, one for $\cX$ and one for $\cY$, are set up so that at each time $t$ they observe loss vectors $\vec{\ell}^t_x \defeq -\mat{A}\vec{y}^t$ and $\vec{\ell}^t_y \defeq \mat{A}^{\!\top}\! \vec{x}^t$, respectively, where $\vec{x}^t \in \cX$ and $\vec{y}^t \in\cY$ are the decisions output by the two regret minimizers. A well-known folk theorem asserts that in doing so, at time $T$ the average decisions $(\bar{\vec{x}}^T, \bar{\vec{y}}^T) \defeq (\frac{1}{T}\sum_{t=1}^T \vec{x}^t, \frac{1}{T}\sum_{t=1}^T \vec{y}^t)$ have \emph{saddle-point gap} (a standard measure of how close a point is to being a saddle-point) $\gamma(\bar{\vec{x}}^T,\bar{\vec{y}}^T)\defeq \max_{\hat{\vec{x}} \in \cX} \hat{\vec{x}}^{\!\top}\!\! \mat{A} \bar{\vec{y}}^T - \min_{\hat{\vec{y}} \in \cY} (\bar{\vec{x}}^T)^{\!\top}\!\! \mat{A} \hat{\vec{y}}$ bounded above by $\gamma(\bar{\vec{x}}^T,\bar{\vec{y}}^T) \le (R^T_\cX + R^T_\cY)/T$ where $R^T_\cX$ and $R^T_\cY$ are the cumulative regrets of the regret minimizers. Since the regrets grow sublinearly, $\gamma(\bar{\vec{x}}^T,\bar{\vec{y}}^T) \to 0$ as $T \to+\infty$. As discussed in the introduction, this approach has been extremely successful in computational game theory.


    \section{Scaled Extension: A Convexity-Preserving Operation for Incrementally Constructing Strategy Spaces}

In this section, we introduce a new convexity-preserving operation between two sets. We show that it provides an alternative way of constructing the strategy space of a player in an extensive-form game that is different from the construction based on convex hulls and Cartesian products described by~\citet{Farina19:Regret}.
Our new construction enables one to \emph{incrementally extend} the strategy space in a top-down fashion, whereas the construction by~\citet{Farina19:Regret} was bottom-up. Most importantly, as we will show in \cref{sec:unrolling sf}, this new operation enables one to incrementally, recursively construct the extensive-form correlated strategy space (again in a top-down fashion). 

\begin{definition}\label{def:scaled extension}
Let $\cX$ and $\cY$ be nonempty, compact and convex sets, and let $h : \cX \to\bbR_+$ be a nonnegative affine real function. The \emph{scaled extension} of $\cX$ with $\cY$ via $h$ is defined as the set
\[
  \cX \ext^h \cY \defeq \{(\vec{x}, \vec{y}) : \vec{x} \in \cX,\ \vec{y} \in h(\vec{x}) \cY\}.
\]
\end{definition}

Since we will be composing multiple scaled extensions together, it is important to verify that the operation above not only preserves convexity, but also preserves the non-emptiness and compactness of the sets (a proof of the following Lemma is available in \cref{app:scaled extension}):

\begin{restatable}{lemma}{lemscaledextension}
  Let $\cX, \cY$ and $h$ be as in \cref{def:scaled extension}. Then $\displaystyle\cX \ext^h \cY$ is nonempty, compact and convex.
\end{restatable}

\subsection{Construction of the Set of Sequence-Form
Strategies}
\label{sec:unrolling sf}

The scaled extension operation can be used to construct the polytope of a perfect-recall player's strategy in sequence-form in an extensive-form game. We illustrate the approach in the small example of \cref{fig:small sf}; the generalization to any extensive-form strategy space is immediate. As noted in \cref{sec:efg}, any valid sequence-form strategy must satisfy probability mass constraints, and can be constructed incrementally in a top-down fashion, as follows (in the following we refer to the same naming scheme as in \cref{fig:small sf} for the sequences of Player 1):
\begin{enumerate}[leftmargin=*,nolistsep,itemsep=0mm,label=\roman*.]
  \item First, the empty sequence is set to value $x[\emptyseq] = 1$.
  \item (Info set \textsc{a}) Next, the value $x[\emptyseq]$ is {partitioned} into the two non-negative values $x[1]\!+\! x[2] \!=\! x[\emptyseq]$.
  \item (Info set \textsc{b}) Next, the value $x[1]$ is partitioned into two non-negative values $x[3] + x[4] = x[1]$.
  \item (Info set \textsc{c}) Next, the value $x[1]$ is partitioned into two non-negative values $x[5] + x[6] = x[1]$.
  \item (Info set \textsc{d}) Next, the value $x[2]$ is partitioned into 3 non-negative values $x[7] \!+\! x[8] \!+\! x[9] \!=\! x[2]$.
\end{enumerate}
The incremental choices in the above recipe can be directly translated---in the same order---into set operations by using scaled extensions, as follows:
\begin{enumerate}[leftmargin=*,nolistsep,itemsep=0mm,label=\roman*.]
  \item First, the set of all feasible values of sequence $x[\emptyseq]$ is the singleton $\cX_0 \defeq \{1\}$.
  \item Then, the set of all feasible values of $(x[\emptyseq], x[1], x[2])$ is the set $ \cX_1 \defeq \cX_0 \times \symp{2} = \cX_0 \ext^{h_1} \symp{2}$, where $h_1$ is the linear function $h_1 : \cX_0 \ni x[\emptyseq] \mapsto x[\emptyseq]$ (the identity function).
  \item In order to characterize the set of all feasible values of $(x[\emptyseq], \dots, x[4])$ we start from $\cX_1$, and \emph{extend} any element $(x[\emptyseq], x[1], x[2]) \in \cX_1$ with the two sequences $x[3]$ and $x[4]$, drawn from the set $\{(x[3], x[4]) \in \bbR_2^+ : x[3] + x[4] = x[1]\} = x[1] \symp{2}$. We can express this extension using scaled extension:
      $\cX_2 \defeq \cX_1 \ext^{h_2} \symp{2}$, where $h_2 : \cX_1 \ni (x[\emptyseq], x[1], x[2]) \mapsto x[1]$.
  \item Similarly, we can extend every element in $\cX_2$ to include $(x[5], x[6]) \in x[1] \symp{2}$: in this case, $\cX_3 \defeq \cX_2 \ext^{h_3} \symp{2}$, where $h_3 : \cX_2 \ni (x[\emptyseq], x[1], x[2], x[3], x[4]) \mapsto x[1]$.
  \item The set of all feasible $\!(x[\emptyseq],..,x[9])$ is $\cX_4 \!\defeq\! \cX_3\! \ext^{h_4}\! \symp{3}$, where $h_4 \!:\! \cX_3 \!\ni\! (x[\emptyseq], \dots,\! x[6]) \!\mapsto\! x[2]$.
\end{enumerate}

Hence, the polytope of sequence-form strategies for Player 1 in \cref{fig:small sf} can be expressed as
\[\displaystyle
  \{1\} \ext^{h_1}\symp{2}\ext^{h_2}\symp{2}\ext^{h_3}\symp{2}\ext^{h_4} \symp{3},
\]
where the scaled extension operation is intended as left associative.

\subsection{Regret Minimizer for Scaled Extension}\label{sec:circuit}

It is always possible to construct a regret minimizer for $\cZ = \displaystyle \cX \ext^h \cY$, where $h(\vec{x}) = \langle\vec{a},\vec{x}\rangle + b$, starting from a regret minimizer for $\cX \subseteq \bbR^{m}$ and $\cY \subseteq \bbR^{n}$. The fundamental technical insight of the construction is that, given any vector $\vec{\ell} = (\vec{\ell}_x, \vec{\ell}_y) \in \bbR^{m} \times \bbR^{n}$, the minimization of a linear function $\vec{z}\mapsto \langle \vec{\ell}, \vec{z}\rangle$ over $\cZ$ can be split into two separate linear minimization problems over $\cX$ and $\cY$:
\begin{align*}
  \min_{\vec{z} \in \cZ} \ \langle \vec{\ell}, \vec{z}\rangle &= \min_{\substack{\vec{x}\in \cX , \vec{y}\in\cY}} \big\{\langle \vec{\ell}_x,\vec{x}\rangle + h(\vec{x})\langle \vec{\ell}_y, \vec{y}\rangle \big\} = \min_{\vec{x} \in \cX} \big\{\langle \vec{\ell}_x, \vec{x} \rangle + h(\vec{x}) \min_{\vec{y}\in\cY} \langle \vec{\ell}_y, \vec{y}\rangle \big\} \\
&= \min_{\vec{x} \in \cX} \big\{ \big\langle \vec{\ell}_x + \vec{a}\cdot \min_{\vec{y}\in\cY} \langle \vec{\ell}_y, \vec{y}\rangle, \vec{x}\big\rangle\big\} + b\cdot\min_{\vec{y}\in\cY}\langle \vec{\ell}_y, \vec{y}\rangle.
\end{align*}
Thus, it is possible to break the problem of minimizing regret over $\cZ$ into two regret minimization subproblems over $\cX$ and $\cY$ (more details in \cref{app:regret circuit}). In particular:
\begin{proposition}\label{prop:regret bound}
Let \regm{X} and \regm{Y} be two regret minimizer over $\cX$ and $\cY$ respectively, and let $R^T_\cX, R^T_\cY$ denote their cumulative regret at time $T$. Then, \cref{algo:rm} provides a regret minimizer over $\cZ$ whose cumulative regret $R^T_\cZ$ is bounded above as $R_\cZ^T \le R_\cX^T + h^\ast R_\cY^T$, where $h^* \defeq \max_{\vec{x}\in\cX} h(\vec{x})$.
\vspace{-3mm}
\begin{algorithm}[H]
    \caption{Regret minimizer over the scaled extension $ \cX \ext^h \cY$.}
    \label{algo:rm}
    \small
    \hspace{1mm}\includegraphics[scale=.83]{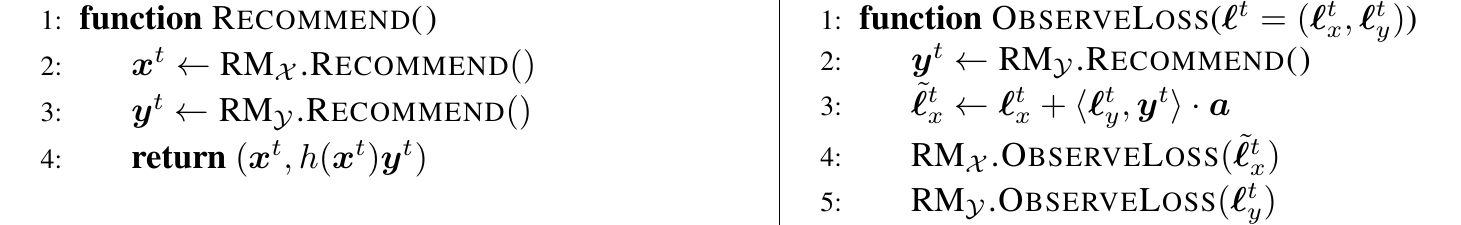}
\end{algorithm}
\end{proposition}
\vspace{-4mm}

\cref{algo:rm} can be composed recursively to construct a regret minimizer for any set that is expressed via a chain of scaled extensions, such as the polytope of sequence-form strategies (\cref{sec:unrolling sf}) or that of extensive-form correlation plans (\cref{sec:unrolling efce}).
When used on the polytope of sequence-form strategies, \cref{algo:rm} coincides with the CFR algorithm if all regret minimizers for the individual simplexes in the chain of scaled extensions are implemented using the regret matching algorithm~\citep{Hart00:Simple}.

    \section{Unrolling the Structure of the Correlated Strategy Polytope}
\label{sec:unrolling efce}

In this section, we study the combinatorial structure of the polytope of correlated strategies (\cref{sec:xi polytope}) of a two-player perfect-recall extensive-form game with no chance moves. The central result of this section, \cref{thm:decomposition}, asserts that the correlated strategy polytope $\Xi$ can be expressed via a chain of scaled extensions. This matches the similar result regarding the sequence-form strategy polytope that we discussed in \cref{sec:unrolling sf}. However, unlike the sequence-form strategy polytope, the constraints that define the correlated strategy polytope do not exhibit a natural hierarchical structure: the constraints that define $\Xi$ (\cref{def:xi}) are such that the same entry of the correlation plan $\vec{\xi}$ can appear in multiple constraints, and furthermore the constraints will in general form cycles. This makes the problem of unrolling the structure of $\Xi$ significantly more challenging.

The key insight is that some of the constraints that define $\Xi$ are redundant (that is, implied by the remaining constraints) and can therefore be safely eliminated. Our algorithm identifies one such set of redundant constraints, and removes them. The set is chosen in such a way that the remaining constraints can be laid down in a hierarchical way that can be captured via a chain of scaled extensions.

\subsection{Example}\label{sec:example unroll}

Before we delve into the technical details of the construction, we illustrate the key idea of the algorithm in a small example. In particular, consider the small game tree of \cref{fig:example unroll} (left), where we used the same conventions as in \cref{sec:efg} and \cref{fig:small sf}. All sequence pairs are relevant; the set of constraints that define $\Xi$ is shown in \cref{fig:example unroll} (middle).

\begin{figure}[ht]
\begin{minipage}[t]{2.5cm}
\begin{tikzpicture}[baseline=0pt]
  \node[fill=black,draw=black,circle,inner sep=.5mm] (A) at (0, 0) {};
  \node[fill=white,draw=black,circle,inner sep=.5mm] (B) at (-.57,-.8) {};
  \node[fill=white,draw=black,circle,inner sep=.5mm] (C) at (.57,-.8) {};
  \node[fill=white,draw=black,inner sep=.6mm] (l1) at (-1,-1.6) {};
  \node[fill=white,draw=black,inner sep=.6mm] (l2) at (-.25,-1.6) {};
  \node[fill=white,draw=black,inner sep=.6mm] (l4) at (1,-1.6) {};
  \node[fill=white,draw=black,inner sep=.6mm] (l3) at (.25,-1.6) {};

  \draw[semithick] (A) --node[fill=white,inner sep=.9] {\scriptsize$1$} (B) --node[fill=white,inner sep=.9] {\scriptsize$1$} (l1);
  \draw[semithick] (B) --node[fill=white,inner sep=.9] {\scriptsize$2$} (l2);
  \draw[semithick] (A) --node[fill=white,inner sep=.9] {\scriptsize$2$} (C) --node[fill=white,inner sep=.9] {\scriptsize$3$} (l3);
  \draw[semithick] (C) --node[fill=white,inner sep=.9] {\scriptsize$4$} (l4);

  \draw[black!60!white] (A) circle (.2);
  \node[black!60!white]  at ($(A) + (-.4, 0)$) {\textsc{a}};

  \draw[black!60!white] (B) circle (.2);
  \node[black!60!white]  at ($(B) + (-.4, 0)$) {\textsc{b}};

  \draw[black!60!white] (C) circle (.2);
  \node[black!60!white]  at ($(C) + (.4, 0)$) {\textsc{c}};
\end{tikzpicture}
\end{minipage}\hspace{4mm}
\begin{minipage}[t]{7.3cm}\small
  In this game, $\Xi$ is defined by the following constraints:\\[1mm]
  $\displaystyle
    \left\{\!\!\begin{array}{lr}
      \xi[\emptyseq,\emptyseq] = 1,\\
      \xi[\sigma_1,1] + \xi[\sigma_1,2] = \xi[\sigma_1, \emptyseq] &\forall \sigma_1\in\{\emptyseq,1,2\},\\
      \xi[\sigma_1,3] + \xi[\sigma_1,4] = \xi[\sigma_1, \emptyseq] &\forall \sigma_1\in\{\emptyseq,1,2\},\\
      \xi[1,\sigma_2] + \xi[2,\sigma_2] = \xi[\emptyseq, \sigma_2] &\forall \sigma_2\in\{\emptyseq,1,2,3,4\}.
    \end{array}\right.\hspace{-3mm}
  $
\end{minipage}\hspace{5mm}
\begin{minipage}[t]{3.1cm}%
  \def\sep{.1}%
  \def\side{.46}%
  \def\dd{.07}%
  \begin{tikzpicture}[baseline=3pt]
    \tikzstyle{outer}=[semithick];
    \begin{scope}
        \clip(-.29,.29) rectangle (5*\side + 2*\sep+0.05,-3*\side-\sep-0.05);
        \draw[outer] (0, 0) rectangle (\side, -\side);
        \draw[outer] ($(\side + \sep,0)$) rectangle ($(3 * \side + \sep, -\side)$);
        \draw[outer] ($(3 * \side + 2 * \sep, 0)$) rectangle ($(5 * \side + 2 * \sep, -\side)$);
        \draw[outer] ($(0, -\side - \sep)$) rectangle ($(\side, -3 * \side - \sep)$);
        \draw[outer] ($(\side + \sep, -\side - \sep)$) rectangle ($(3 * \side + \sep, -3 * \side - \sep)$);
        \draw[outer] ($(3 * \side + 2 * \sep, -\side - \sep)$) rectangle ($(5 * \side + 2 * \sep, -3 * \side - \sep)$);

        \draw ($(0, -2*\side -\sep)$) -- ($(\side,-2*\side-\sep)$);
        \draw ($(\side+\sep, -2*\side -\sep)$) -- ($(3*\side+\sep,-2*\side-\sep)$);
        \draw ($(3*\side+2*\sep, -2*\side -\sep)$) -- ($(5*\side+2*\sep,-2*\side-\sep)$);
        \draw ($(2*\side+\sep, 0)$) -- ($(2*\side+\sep,-\side)$);
        \draw ($(2*\side+\sep, -\side-\sep)$) -- ($(2*\side+\sep,-3*\side-\sep)$);
        \draw ($(4*\side+2*\sep, 0)$) -- ($(4*\side+2*\sep,-\side)$);
        \draw ($(4*\side+2*\sep, -\side-\sep)$) -- ($(4*\side+2*\sep,-3*\side-\sep)$);

        \node at ($(.5*\side,.2)$) {\scriptsize$\emptyseq$};
        \node at ($(\sep+1.5*\side,.2)$) {\scriptsize$1$};
        \node at ($(\sep+2.5*\side,.2)$) {\scriptsize$2$};
        \node at ($(2*\sep+3.5*\side,.2)$) {\scriptsize$3$};
        \node at ($(2*\sep+4.5*\side,.2)$) {\scriptsize$4$};

        \node at ($(-.2,-.5*\side)$) {\scriptsize$\emptyseq$};
        \node at ($(-.2,-\sep-1.5*\side)$) {\scriptsize$1$};
        \node at ($(-.2,-\sep-2.5*\side)$) {\scriptsize$2$};

        \fill[black!10!white,opacity=.9] ($(\dd,-\dd)$) rectangle node[black!60!white,opacity=1]{\circled{1}} ($(\side-\dd,-\side+\dd)$);
        \fill[black!10!white,opacity=.9] ($(\dd,-\side-\sep-\dd)$) rectangle node[black!60!white,opacity=1]{\circled{2}} ($(\side-\dd,-3*\side-\sep+\dd)$);
        \fill[black!10!white,opacity=.9] ($(\side+\sep+\dd,-\side-\sep-\dd)$) rectangle node[black!60!white,opacity=1]{\circled{3}} ($(3*\side+\sep-\dd,-3*\side-\sep+\dd)$);
        \fill[black!10!white,opacity=.9] ($(3*\side+2*\sep+\dd,-\side-\sep-\dd)$) rectangle node[black!60!white,opacity=1]{\circled{3}} ($(5*\side+2*\sep-\dd,-3*\side-\sep+\dd)$);
        \fill[black!10!white,opacity=.9] ($(\side+\sep+\dd,-\dd)$) rectangle node[black!60!white,opacity=1]{\circled{4}} ($(3*\side+\sep-\dd,-\side+\dd)$);
        \fill[black!10!white,opacity=.9] ($(3*\side+2*\sep+\dd,-\dd)$) rectangle node[black!60!white,opacity=1]{\circled{4}} ($(5*\side+2*\sep-\dd,-\side+\dd)$);
    \end{scope}
  \end{tikzpicture}
\end{minipage}
\vspace{-1mm}
\caption{\small(Left) Example game (\cref{sec:example unroll}). (Middle) Constraints that define $\Xi$ in the example game. (Right) Fill-in order of $\vec{\xi}$. The cell at the intersection of row $\sigma_1$ and column $\sigma_2$ represents the entry $\xi[\sigma_1,\sigma_2]$ of $\vec{\xi}$.}
\label{fig:example unroll}
\vspace{-2mm}
\end{figure}
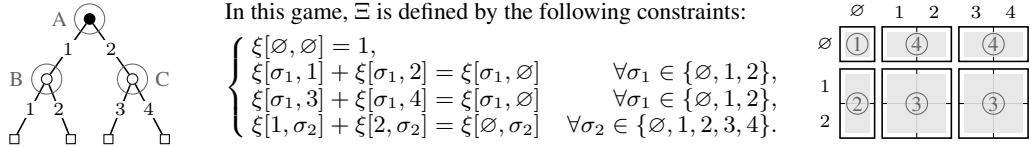

In order to generate all possible correlation plans $\vec{\xi} \in \Xi$, we proceed as follows. First, we assign $\xi[\emptyseq,\emptyseq] = 1$. Then, we partition $\xi[\emptyseq,\emptyseq]$ into two non-negative values $(\xi[1,\emptyseq],\xi[2,\emptyseq]) \in \xi[\emptyseq,\emptyseq]\symp{2}$ in accordance with the constraint $\xi[1,\emptyseq] + \xi[2,\emptyseq] = \xi[\emptyseq,\emptyseq]$. Next, using the constraints $\xi[\sigma_1, 1] + \xi[\sigma_1, 2] = \xi[\sigma_1, \emptyseq]$ and $\xi[\sigma_1, 3] + \xi[\sigma_1, 4] = \xi[\sigma_1, \emptyseq]$, we pick values $(\xi[\sigma_1, 1],\xi[\sigma_1,2])\in\xi[\sigma_1,\emptyseq]\symp{2}$ and $(\xi[\sigma_1, 3], \xi[\sigma_1, 4]) \in \xi[\sigma_1,\emptyseq]\symp{2}$ for $\sigma_1\in\{1,2\}$. So far, our strategy for filling the correlation plan has been to \emph{split} entries according to the information structure of the players. As shown in \cref{sec:unrolling sf}, these steps can be expressed via scaled extension operations.

Next, we fill in the four remaining entries in $\xi$, that is $\xi[\emptyseq,\sigma_2]$ for $\sigma_2\in\{1,2,3,4\}$, in accordance with constraint $\xi[1,\sigma_2] + \xi[2,\sigma_2] = \xi[\emptyseq,\sigma_2]$. In this step, we are not splitting any value; rather, we fill in $\xi[\emptyseq,\sigma_2]$ in the only possible way (that is, $\xi[\emptyseq,\sigma_2] = \xi[1,\sigma_2] + \xi[2,\sigma_2]$), by means of a linear combination of already-filled-in entries. This operation
can be also expressed via scaled extensions, with the singleton set $\{1\}$:
$
  \{(\xi[1,\sigma_2], \xi[2,\sigma_2], \xi[\emptyseq,\sigma_2])\} = \{(\xi[1,\sigma_2], \xi[2,\sigma_2])\} \ext^h \{1\},
$
where $h: (\xi[1,\sigma_2], \xi[2,\sigma_2])\mapsto \xi[1,\sigma_2] + \xi[2,\sigma_2]$ (note that $h$ respects the requirements of \cref{def:scaled extension}). This way, we have filled in all entries in $\xi$. However, only 9 out of the 11 constraints have been taken into account in the construction, and we still need to verify that the two leftover constraints $\xi[\emptyseq,1] + \xi[\emptyseq,2] = \xi[\emptyseq,\emptyseq]$ and $\xi[\emptyseq,3] + \xi[\emptyseq, 4] = \xi[\emptyseq,\emptyseq]$ are automatically satisfied by our way of filling in the entries of $\vec{\xi}$. Luckily, this is always the case: by construction, $\xi[\emptyseq,1] \!+\! \xi[\emptyseq,2] \!=\! (\xi[1,1] \!+\! \xi[1,2]) \!+\! (\xi[2,1] \!+\! \xi[2,2]) \!=\! \xi[1,\emptyseq] \!+\! \xi[2,\emptyseq] \!=\! \xi[\emptyseq,\emptyseq]$ (the proof for $\xi[\emptyseq,3] + \xi[\emptyseq,4]$ is analogous). %
We summarize the construction steps pictorially in \cref{fig:example unroll} (right).

\begin{remark}\label{rem:unfavorable}
    Similar construction that starts from assigning values for $\xi[\emptyseq,\sigma_2]$ ($\sigma_2\in\{1,2,3,4\}$ using constraints $\xi[\emptyseq,1]+\xi[\emptyseq,2]=\xi[\emptyseq,\emptyseq]$, $\xi[\emptyseq,3]+\xi[\emptyseq,4]=\xi[\emptyseq,\emptyseq]$ and fills out $\xi[\sigma_1,\sigma_2]$ for $(\sigma_1,\sigma_2)\in\{1,2\}\times\{1,2,3,4\}$ would have not been successful: if $(\xi[1,1],\xi[1,2])$ and $(\xi[1,3],\xi[1,4])$ are filled in independently, there is no way of guaranteeing that $\xi[1,1]+\xi[1,2] = \xi[1,3]+\xi[1,4]$ ($=\xi[1,\emptyseq]$) as required by the constraints.
\end{remark}

\subsection{An Unfavorable Case that Cannot Happen in Games with No Chance Moves}\label{sec:unfavorable}

We now show that there exist game instances in which the general approach used in the previous subsection fails.
In particular, consider a relevant sequence pair $(\sigma_1,\sigma_2)$ such that both $\sigma_1$ and $\sigma_2$ are parent sequences of two information sets of Player 1 and Player 2 respectively, and assume that all sequence pairs in the game are relevant. Then, no matter what the order of operations is, the situation described in \cref{rem:unfavorable} cannot be avoided. Luckily, in two-player perfect-recall games with no chance moves, one can prove that this occurrence never happens (see \cref{app:unfavorable} for a proof):

\begin{restatable}{proposition}{propunfavorable}\label{prop:unfavorable}
  Consider a two-player perfect-recall game with no chance moves, and let $(\sigma_1,\sigma_2)$ be a relevant sequence pair, let $I_1, I_1'$ be two distinct information sets of Player 1 such that $\sigma(I_1) = \sigma(I_1') = \sigma_1$, and let $I_2,I_2'$ be two distinct information sets of Player 2 such that $\sigma(I_2) = \sigma(I_2') = \sigma_2$. It is not possible that both $I_1 \conn I_2$ and $I_1' \conn I_2'$.
\end{restatable}

In other words, if $I_1 \conn I_2$, then any pair of sequences $(\sigma'_1,\sigma'_2)$ where $\sigma'_1$ belongs to $I_1'$ and $\sigma'_2$ belongs to $I_2'$ is \emph{irrelevant}. As we show in the next subsection, this is enough to yield a polynomial-time algorithm to `unroll' the process of filling in the entries of $\vec{\xi} \in \Xi$ in any two-player perfect-recall extensive-form game with no chance moves.
The following definition is crucial for that algorithm:
\begin{definition}
  Let $(\sigma_1,\sigma_2)$ be a relevant sequence pair, and let $I_1 \in \mathcal{I}_1$ be an information set for Player 1  such that $\sigma(I_1) = \sigma_1$. Information set $I_1$ is called \emph{critical for} $\sigma_2$ if there exists at least one $I_2 \in \mathcal{I}_2$ with $\sigma(I_2) = \sigma_2$ such that $I_1 \conn I_2$. (A symmetric definition holds for an $I_2 \in \mathcal{I}_2$.)
\end{definition}
It is a simple corollary of \cref{prop:unfavorable} that for any relevant sequence pair, \emph{at least one} player has \emph{at most one} critical information set for the opponent's sequence. We call such a player \emph{critical} for that relevant sequence pair.

\subsection{A Polynomial-Time Algorithm that Decomposes $\Xi$ using Scaled Extensions}\label{sec:unroll algo}

In this section, we present the central result of the paper: an efficient algorithm that expresses $\Xi$ as a chain of scaled extensions of simpler sets. In particular, as we have already seen in \cref{sec:example unroll}, each set in the decomposition is either a simplex (when \emph{splitting} an already-filled-in entry) or the singleton set $\{1\}$ (when \emph{summing} already filled-in entries and assigning the result to a new entry of $\vec{\xi}$).

The algorithm consists of a recursive function, $\textsc{Decompose}$, which takes three arguments: a relevant sequence pair $(\sigma_1, \sigma_2)$, a subset $\mathcal{S}$ of the set of all relevant sequence pairs, and a set $\mathcal{D}$ of vectors with entries indexed by the elements in $\mathcal{S}$. $\mathcal{S}$ represents the set of indices of $\vec{\xi}$ that have already been filled in, while $\mathcal{D}$ is the set of all partially-filled-in correlation plans (see \cref{sec:example unroll}). The decomposition for the whole polytope $\Xi$ is obtained by evaluating $\textsc{Decompose}((\emptyseq,\emptyseq), \mathcal{S} = \{(\emptyseq,\emptyseq)\}, \mathcal{D} = \{ (1) \})$, which corresponds to the starting situation in which only the entry $\xi[\emptyseq,\emptyseq]$ has been filled in (with the value 1 as per \cref{def:xi}). Each call to \textsc{Decompose} returns a pair $(\mathcal{S}',\mathcal{D}')$ of updated indices and partial vectors, to reflect the new entries that were filled in during the call.
Each call to $\textsc{Decompose}((\sigma_1, \sigma_2), \mathcal{S}, \mathcal{D})$ works as follows:
\begin{itemize}[leftmargin=*,nolistsep,itemsep=0mm]
  \item First, the algorithm finds one critical player for the relevant sequence pair $(\sigma_1,\sigma_2)$ (see end of \cref{sec:unfavorable}). Assume without loss of generality that Player 1 is critical (the other case is symmetric), and let $\mathcal{I}^* \subseteq \mathcal{I}_1$ be the set of critical information sets for $\sigma_2$ that belong to Player 1. By definition of critical player, $\mathcal{I}^*$ is either a singleton or it is an empty set.
  \item For each $I \in \mathcal{I}_1$ such that $\sigma(I) = \sigma_1$ and $I \rele \sigma_2$, we:
       \begin{itemize}[nolistsep,itemsep=0mm]
           \item Fill in all entries $\{\xi[(I^*,a),\sigma_2]:a\in A_{I}\}$ by splitting $\xi[\sigma_1,\sigma_2]$. This is reflected by updating the set of filled-in-indices $\mathcal{S} \gets \mathcal{S} \cup \{((I,a), \sigma_2)\}$ and extending $\mathcal{D}$ via a scaled extension: $\mathcal{D} \gets \mathcal{D} \ext^h \symp{|A_{I}|}$ where $h$
               extracts $\xi[\sigma_1,\sigma_2]$ from any partially-filled-in vector.
           \item Then, for each $a\in A_{I}$ we assign $(\mathcal{S}, \mathcal{D}) \gets \textsc{Decompose}(((I, a), \sigma_2), \mathcal{S}, \mathcal{D})$.
       \end{itemize}
       After this step, all the indices in $\{(\sigma_1',\sigma_2'): \sigma_1' \succ \sigma_1,\sigma'_2 \succeq \sigma_2\} \cup \{(\sigma_1,\sigma_2)\}$ have been filled in, and none of the indices in $\{(\sigma_1, \sigma'_2): \sigma'_2 \succ \sigma_2\}$ have been filled in yet.
    \item Finally, we fill out all indices in $\{(\sigma_1, \sigma'_2): \sigma'_2 \succ \sigma_2\}$. We do so by iterating over all information sets $J \in \mathcal{I}_2$ such that $\sigma(J) \succeq \sigma_2$ and $\sigma_1 \rele J$. For each such $J$, we split into two cases, according to whether $\mathcal{I}^* = \{I^*\}$ (for some $I^*\in \mathcal{I}_1$, as opposed to $\mathcal{I}^*$ being empty) and $J \conn I^*$, or not:
        \begin{itemize}[nolistsep,itemsep=0mm]
            \item If $\mathcal{I}^* = \{I^*\}$ and $J \conn I^*$, then for all $a \in A_J$ we fill in the sequence pair $\xi[\sigma_1, (J,a)]$ by assigning its value in accordance with the constraint $\xi[\sigma_1, (J,a)] = \sum_{a^* \in A_{I^*}} \xi[(I^*, a^*), (J,a)]$ via the scaled extension $\mathcal{D} \gets \mathcal{D} \ext^h \{ 1\}$ where the linear function $h$ maps a partially-filled-in vector to the value of $\sum_{a^* \in A_{I^*}} \xi[(I^*, a^*), (J,a)]$. 
            \item Otherwise, we fill in the entries $\{\xi[\sigma_1, (J,a)] : a\in A_J\}$, by splitting the value $\xi[\sigma_1, \sigma(J)]$. In other words, we let
                $\mathcal{D} \gets \mathcal{D} \ext^h \symp{|A_J|}$ where $h$ extracts the entry $\xi[\sigma_1, \sigma(J)]$ from a partially-filled-in vector in $\mathcal{D}$.
        \end{itemize}
        \item At this point, all the entries corresponding to indices $\tilde{\mathcal{S}} = \{(\sigma'_1, \sigma'_2): \sigma'_1 \succeq \sigma_1, \sigma'_2 \succeq \sigma_2\}$ have been filled in, and we return $(\mathcal{S} \cup \tilde{\mathcal{S}}, \mathcal{D})$.
\end{itemize}
  Every call to \textsc{Decompose} increases the cardinality of $\mathcal{S}$ by at least one unit. Since $\mathcal{S}$ is a subset of the set of relevant sequence pairs, and since the total number of relevant sequence pair is polynomial in the input game tree size, the algorithm runs in polynomial time. See \cref{app:algo} for pseudocode, as well as a proof of correctness of the algorithm.
   Since every change to $\mathcal{D}$ is done via scaled extensions (with either a simplex or the singleton set $\{1\}$), we conclude that:
\begin{theorem}\label{thm:decomposition}
  In a two-player perfect-recall EFG with no chance moves, the space of correlation plans $\Xi$ can be expressed via a sequence of scaled extensions with simplexes and singleton sets:
\begin{equation}\label{eq:xi decomposition}
  \Xi = \{1\} \ext^{h_1} \cX_1 \ext^{h_2} \cX_2 \ext^{h_3} \cdots \ext^{h_n} \cX_n, \text{ where, for } i = 1,\dots, n,\text{ either } \cX_i = \Delta^{s_i}\text{ or } \cX_i = \{1\},
\end{equation}
and $h_i(\cdot) = \langle \vec{a_i},\cdot\rangle$ is a linear function.
 Furthermore, an exact algorithm exists to compute such expression in polynomial time.
\end{theorem}
 We can recursively use \cref{algo:rm} on the expression~\eqref{eq:xi decomposition} to obtain a regret minimizer for $\Xi$. The resulting algorithm, shown in \cref{algo:full rm} of \cref{app:full rm algo}, is contingent on a choice of ``local'' regret minimizers $\text{RM}_i$ for each of the simplex domains $\Delta^{s_i}$ in~\eqref{eq:xi decomposition}. By virtue of~\cref{algo:rm}, if each local regret minimizer $\text{RM}_i$ for $\Delta^{s_i}$ runs in linear time (i.e., computes recommendations and observes losses by running an algorithm whose complexity is linear in $s_i$)\footnote{Linear-time regret minimizers for simplexes include regret-matching~\citep{Hart00:Simple}, regret-matching-plus~\citep{Tammelin15:Solving}, mirror-descent and follow-the-regularized-leader (e.g, \citet{Hazan16:Introduction}).}, then the overall regret minimization algorithm for $\Xi$ runs in linear time in the number of relevant sequence pairs of the game. Furthermore, \cref{prop:regret bound} immediately implies that if each $\text{RM}_i$ is Hannan consistent, then so is our overall algorithm for $\Xi$. Putting these observations together, we conclude:

\begin{restatable}{theorem}{thmrmproperties}
  For any two-player extensive-form game with no chance moves, there exists a Hannan consistent regret minimizer for $\Xi$ that runs in linear time in the number of relevant sequence pairs.
\end{restatable} 
    \section{Experimental Evaluation}

We experimentally evaluate the scalability of our regret-minimization algorithm for computing an extensive-form correlated equilibrium. In particular, we implement a regret minimizer for the space of correlation plans by computing the structural decomposition of $\Xi$ into a chain of scaled extensions (\cref{sec:unroll algo}) and repeatedly applying the construction of \cref{sec:circuit}. This regret minimizer is then used on the saddle-point formulation of an EFCE (\cref{sec:efce}) as explained in \cref{sec:regret}, with two modifications that are standard in the literature on regret minimization algorithms for game theory~\citep{Tammelin15:Solving,Burch19:Revisiting}: (i) alternating updates and (ii) linear averaging of the iterates\footnote{The linear average of $n$ vectors $\vec{\xi}_1,\dots,\vec{\xi}_n$ is $(\sum_{t=1}^n t\cdot \vec{\xi}_t)/(\sum_{t=1}^n t) = 2(\sum_{t=1}^n t\cdot \vec{\xi}_t)/(n(n+1))$.}.
\begin{wraptable}{r}{5.5cm}
        \centering\fontsize{7}{9}\selectfont
        \setlength\tabcolsep{1.15mm}
        \begin{tabular}{ccc|ccc}
          \toprule
          Board & {Num} & {Ship} & \multirow{2}{*}{$|\Sigma_1|$} &\multirow{2}{*}{$|\Sigma_2|$}&Num. rel.\\[-.5mm]
          size&turns&length&&&seq. pairs\\
          \midrule
          (3, 2) & $3$ & 1 & 15k & 47k & 3.89M\\
          (3, 2) & $4$ & 1  & 145k & 306k & 26.4M\\
          (3, 2) & $4$ & 2  & 970k & 2.27M  &    111M\\
          \bottomrule
        \end{tabular}
        \caption{\small Game metrics for the different instances of the Battleship game we test on.}
        \label{tab:game sizes}
        \vspace{-3.5mm}
\end{wraptable}
We use \emph{regret-matching-plus}~\citep{Tammelin15:Solving} to minimize the regret over the simplex domains in the structural decomposition. These variants are known to be beneficial in the case of Nash equilibrium, and we observed the same for EFCE. We compare our algorithm to two known algorithms in the literature. The first is based on linear programming~\citep{Stengel08:Extensive}.
The second is a very recent subgradient descent algorithm for this problem~\citep{Farina19:Correlation}, which leverages a recent subgradient descent technique~\citep{Wang13:Incremental}. All algorithms were run on a machine with 16 GB of RAM and an Intel i7 processor with 8 cores. We used the Gurobi commercial solver (while allowing it to use any number of threads) to solve the LP when evaluating the scalability of the LP-based method proposed by~\citet{Stengel08:Extensive}.

\textbf{Game instances.} We test the scalability of our algorithm in a benchmark game for EFCE that was recently proposed by~\citet{Farina19:Regret}: a parametric variant of the classical war game \emph{Battleship}. \cref{tab:game sizes} shows some statistics about the three game instances that we use, including the number of relevant sequence pairs in the game (\cref{def:relevant seq pair}). `Board size' refers to the size of the Battleship playfield; each player has a field of that size in which to place his ship. `Num turns' refers to the maximum number of shots that each player can take (in turns). `Ship length' is the length of the one ship that each player has. Despite the seemingly small board sizes and the presence of only one ship per player, the game trees for these instances are quite large, with each player having tens of thousands to millions of sequences.

\begin{figure}[t]
    \centering
    \includegraphics[scale=.84]{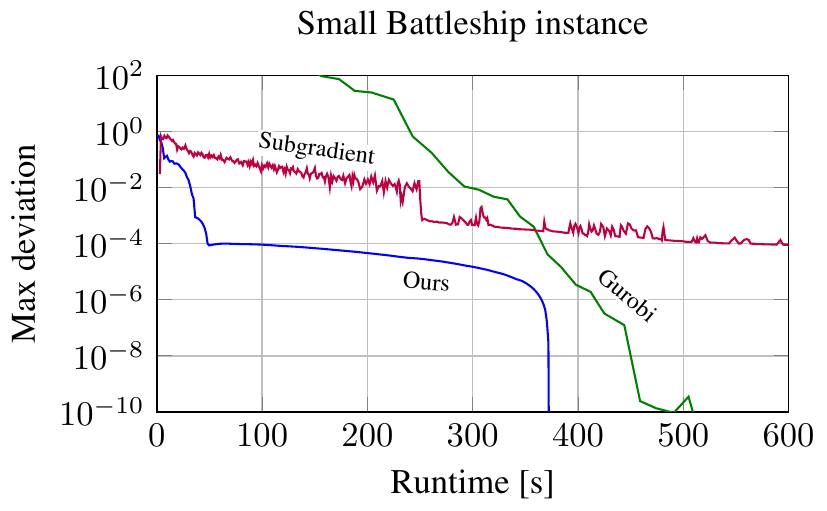}\hfill
    \includegraphics[scale=.84]{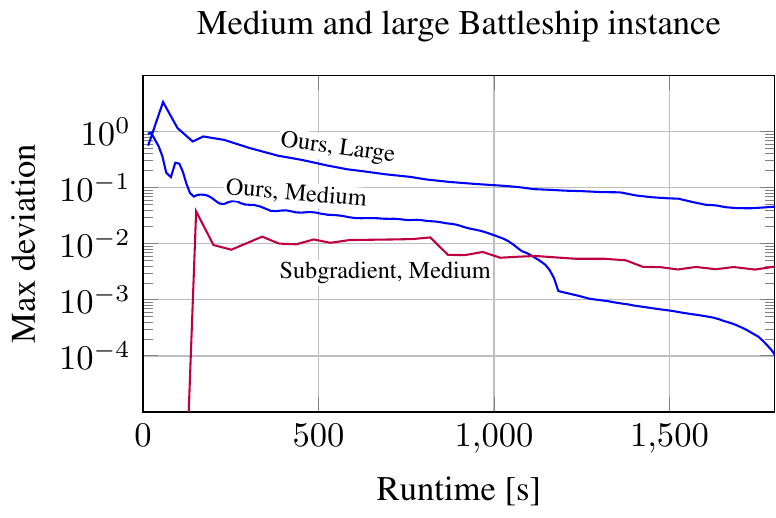}\vspace{-3mm}
    \caption{\small Experimental results. The y-axis shows the maximum utility increase upon deviation.}
    \label{fig:results}
\end{figure}

\textbf{Scalability of the Linear Programming Approach~\citep{Stengel08:Extensive}.} Only the small instance could be solved by Gurobi, \cref{fig:results} (left). (Out of the LP algorithms provided by Gurobi, the barrier method was faster than the primal- and dual-simplex methods.) On the medium and large instance, Gurobi was killed by the system for trying to allocate too much memory.~\citet{Farina19:Correlation} report that the large instance needs more than 500GB of memory in order for Gurobi to run. The Gurobi run time shown in \cref{fig:results} does not include the time needed to construct and destruct the Gurobi LP objects, which is negligible.

\textbf{Scalability of the Very Recent Subgradient Technique~\citep{Farina19:Correlation}.}
The very recent subgradient descent algorithm for this problem was able to solve the small and medium instances if the algorithm's step size was tuned well. An advantage of our technique is that it has no parameters to tune. Another issue is that the iterates $\Xi$ of the subgradient algorithm are not feasible while ours are. Furthermore, on the large instance, the subgradient technique was already essentially unusable because each iteration took over an hour (mainly due to computing the projection).

\cref{fig:results} shows the experimental performance of the subgradient descent algorithm. We used a step size of $10^{-3}$ in the small instance and of $10^{-6}$ in the medium instance. Since the iterates produced by the subgradient technique are not feasible, extra care has to be taken when comparing the performance of the subgradient method to that of our approach or Gurobi. Figure~\ref{fig:subgradient infeasibility} in \cref{app:subgradient} reports the infeasibility of the iterates produced by the subgradient technique over time.

\textbf{Scalability of Our Approach.} We implemented the structural decomposition algorithm of \cref{sec:unroll algo}. Our parallel implementation using 8 threads has a runtime of 2 seconds on the small instance, 6 seconds on the medium instance, and 40 seconds on the large instance (each result was averaged over 10 runs). Finally, we evaluated the performance of the regret minimizer constructed according to~\cref{sec:circuit}; the results are in \cref{fig:results} (left) for the small instance and \cref{fig:results} (right) for the medium and large instance. The plots do not include the time needed to construct and destruct the regret minimizers in memory, which again is negligible. As expected, on the small instance, the rate of convergence of our regret minimizer (a first-order method) is slower than that of the barrier method (a second-order method). However, the barrier method incurs a large overhead at the beginning, since Gurobi spends time factorizing the constraint matrix and computing a good ordering of variables for the elimination tree. The LP-based approach could not solve the medium or large instance, while ours could.  Even on the largest instance, no more than 2GB of memory was reserved by our algorithm.

    \section{Conclusions}

We introduced the
      first efficient regret minimization algorithm for finding an extensive-form
      correlated equilibrium in large two-player general-sum games with no
      chance moves. This is more challenging than designing an algorithm for Nash
      equilibrium because the constraints that define the
      space of correlation plans lack the hierarchical structure of
      sequential strategy spaces and might even form cycles. We showed that some of the constraints are redundant and can be excluded from consideration, and presented an efficient algorithm that generates the space of extensive-form correlation plans incrementally from the remaining constraints. We achieved this decomposition via a special
      convexity-preserving operation that we coined \emph{scaled extension}.
      We showed that a regret minimizer can be designed for a scaled extension of any two convex sets, and that from the decomposition we then obtain a global regret minimizer.
      Our algorithm produces feasible iterates. Experiments showed that it significantly outperforms prior approaches---the LP-based approach and a very recent subgradient descent algorithm---and for larger problems it is the only viable option. 

    \section*{Acknowledgments}
    This material is based on work supported by the National
    Science Foundation under grants IIS-1718457, IIS-1617590,
    and CCF-1733556, and the ARO under award W911NF-17-1-0082. Gabriele Farina is supported by a Facebook fellowship. Co-authors Ling and Fang are supported in part by a research grant from Lockheed Martin.

    \bibliography{dairefs/dairefs}

\begin{thebibliography}{31}
\providecommand{\natexlab}[1]{#1}
\providecommand{\url}[1]{\texttt{#1}}
\expandafter\ifx\csname urlstyle\endcsname\relax
  \providecommand{\doi}[1]{doi: #1}\else
  \providecommand{\doi}{doi: \begingroup \urlstyle{rm}\Url}\fi

\bibitem[Ashlagi et~al.(2008)Ashlagi, Monderer, and
  Tennenholtz]{Ashlagi08:Value}
Ashlagi, I., Monderer, D., and Tennenholtz, M.
\newblock On the value of correlation.
\newblock \emph{Journal of Artificial Intelligence Research}, 33:\penalty0
  575--613, 2008.

\bibitem[Aumann(1974)]{Aumann74:Subjectivity}
Aumann, R.
\newblock Subjectivity and correlation in randomized strategies.
\newblock \emph{Journal of Mathematical Economics}, 1:\penalty0 67--96, 1974.

\bibitem[Bowling et~al.(2015)Bowling, Burch, Johanson, and
  Tammelin]{Bowling15:Heads}
Bowling, M., Burch, N., Johanson, M., and Tammelin, O.
\newblock Heads-up limit hold'em poker is solved.
\newblock \emph{Science}, 347\penalty0 (6218), January 2015.

\bibitem[Brown \& Sandholm(2017{\natexlab{a}})Brown and Sandholm]{Brown17:Safe}
Brown, N. and Sandholm, T.
\newblock Safe and nested subgame solving for imperfect-information games.
\newblock In \emph{Proceedings of the Annual Conference on Neural Information
  Processing Systems (NIPS)}, pp.\  689--699, 2017{\natexlab{a}}.

\bibitem[Brown \& Sandholm(2017{\natexlab{b}})Brown and
  Sandholm]{Brown17:Superhuman}
Brown, N. and Sandholm, T.
\newblock Superhuman {AI} for heads-up no-limit poker: {Libratus} beats top
  professionals.
\newblock \emph{Science}, pp.\  eaao1733, Dec. 2017{\natexlab{b}}.

\bibitem[Brown \& Sandholm(2019{\natexlab{a}})Brown and
  Sandholm]{Brown19:Solving}
Brown, N. and Sandholm, T.
\newblock Solving imperfect-information games via discounted regret
  minimization.
\newblock In \emph{AAAI Conference on Artificial Intelligence (AAAI)},
  2019{\natexlab{a}}.

\bibitem[Brown \& Sandholm(2019{\natexlab{b}})Brown and
  Sandholm]{Brown19:Superhuman}
Brown, N. and Sandholm, T.
\newblock Superhuman {AI} for multiplayer poker.
\newblock \emph{Science}, 365\penalty0 (6456):\penalty0 885--890,
  2019{\natexlab{b}}.
\newblock ISSN 0036-8075.
\newblock \doi{10.1126/science.aay2400}.
\newblock URL \url{https://science.sciencemag.org/content/365/6456/885}.

\bibitem[Brown et~al.(2017)Brown, Kroer, and Sandholm]{Brown17:Dynamic}
Brown, N., Kroer, C., and Sandholm, T.
\newblock Dynamic thresholding and pruning for regret minimization.
\newblock In \emph{AAAI Conference on Artificial Intelligence (AAAI)}, 2017.

\bibitem[Burch et~al.(2019)Burch, Moravcik, and Schmid]{Burch19:Revisiting}
Burch, N., Moravcik, M., and Schmid, M.
\newblock Revisiting {CFR}+ and alternating updates.
\newblock \emph{Journal of Artificial Intelligence Research}, 64:\penalty0
  429--443, 2019.

\bibitem[Davis et~al.(2019)Davis, Waugh, and Bowling]{Davis19:Solving}
Davis, T., Waugh, K., and Bowling, M.
\newblock Solving large extensive-form games with strategy constraints.
\newblock In \emph{AAAI Conference on Artificial Intelligence (AAAI)}, 2019.

\bibitem[Farina et~al.(2017)Farina, Kroer, and Sandholm]{Farina17:Regret}
Farina, G., Kroer, C., and Sandholm, T.
\newblock Regret minimization in behaviorally-constrained zero-sum games.
\newblock In \emph{International Conference on Machine Learning (ICML)}, 2017.

\bibitem[Farina et~al.(2019{\natexlab{a}})Farina, Kroer, and
  Sandholm]{Farina19:Online}
Farina, G., Kroer, C., and Sandholm, T.
\newblock Online convex optimization for sequential decision processes and
  extensive-form games.
\newblock In \emph{AAAI Conference on Artificial Intelligence (AAAI)},
  2019{\natexlab{a}}.

\bibitem[Farina et~al.(2019{\natexlab{b}})Farina, Kroer, and
  Sandholm]{Farina19:Regret}
Farina, G., Kroer, C., and Sandholm, T.
\newblock Regret circuits: Composabilty of regret minimizers.
\newblock In \emph{International Conference on Machine Learning (ICML)},
  2019{\natexlab{b}}.

\bibitem[Farina et~al.(2019{\natexlab{c}})Farina, Ling, Fang, and
  Sandholm]{Farina19:Correlation}
Farina, G., Ling, C.~K., Fang, F., and Sandholm, T.
\newblock Correlation in extensive-form games: Saddle-point formulation and
  benchmarks.
\newblock In \emph{Proceedings of the Annual Conference on Neural Information
  Processing Systems (NeurIPS)}, 2019{\natexlab{c}}.

\bibitem[Gordon et~al.(2008)Gordon, Greenwald, and Marks]{Gordon08:No}
Gordon, G.~J., Greenwald, A., and Marks, C.
\newblock No-regret learning in convex games.
\newblock In \emph{Proceedings of the 25\textsuperscript{th} international
  conference on Machine learning}, pp.\  360--367. ACM, 2008.

\bibitem[Hart \& Mas-Colell(2000)Hart and Mas-Colell]{Hart00:Simple}
Hart, S. and Mas-Colell, A.
\newblock A simple adaptive procedure leading to correlated equilibrium.
\newblock \emph{Econometrica}, 68:\penalty0 1127--1150, 2000.

\bibitem[Hazan(2016)]{Hazan16:Introduction}
Hazan, E.
\newblock Introduction to online convex optimization.
\newblock \emph{Foundations and Trends in Optimization}, 2\penalty0
  (3-4):\penalty0 157--325, 2016.

\bibitem[Huang(2011)]{Huang11:Equilibrium}
Huang, W.
\newblock \emph{Equilibrium computation for extensive games}.
\newblock PhD thesis, London School of Economics and Political Science, January
  2011.

\bibitem[Huang \& von Stengel(2008)Huang and von Stengel]{Huang08:Computing}
Huang, W. and von Stengel, B.
\newblock Computing an extensive-form correlated equilibrium in polynomial
  time.
\newblock In \emph{International Workshop On Internet And Network Economics
  (WINE)}, pp.\  506--513. Springer, 2008.

\bibitem[Jiang \& Leyton-Brown(2015)Jiang and Leyton-Brown]{Jiang15:Polynomial}
Jiang, A.~X. and Leyton-Brown, K.
\newblock Polynomial-time computation of exact correlated equilibrium in
  compact games.
\newblock \emph{Games and Economic Behavior}, 91:\penalty0 347--359, 2015.

\bibitem[Koller et~al.(1996)Koller, Megiddo, and {von
  Stengel}]{Koller96:Efficient}
Koller, D., Megiddo, N., and {von Stengel}, B.
\newblock Efficient computation of equilibria for extensive two-person games.
\newblock \emph{Games and Economic Behavior}, 14\penalty0 (2), 1996.

\bibitem[Morav{\v c}{\'\i}k et~al.(2017)Morav{\v c}{\'\i}k, Schmid, Burch,
  Lis{\'y}, Morrill, Bard, Davis, Waugh, Johanson, and
  Bowling]{Moravvcik17:DeepStack}
Morav{\v c}{\'\i}k, M., Schmid, M., Burch, N., Lis{\'y}, V., Morrill, D., Bard,
  N., Davis, T., Waugh, K., Johanson, M., and Bowling, M.
\newblock Deepstack: Expert-level artificial intelligence in heads-up no-limit
  poker.
\newblock \emph{Science}, 356\penalty0 (6337), May 2017.

\bibitem[Papadimitriou \& Roughgarden(2008)Papadimitriou and
  Roughgarden]{Papadimitriou08:Computing}
Papadimitriou, C.~H. and Roughgarden, T.
\newblock Computing correlated equilibria in multi-player games.
\newblock \emph{Journal of the ACM}, 55\penalty0 (3):\penalty0 14, 2008.

\bibitem[Romanovskii(1962)]{Romanovskii62:Reduction}
Romanovskii, I.
\newblock Reduction of a game with complete memory to a matrix game.
\newblock \emph{Soviet Mathematics}, 3, 1962.

\bibitem[Shalev-Shwartz \& Singer(2007)Shalev-Shwartz and
  Singer]{Shalev07:Primal}
Shalev-Shwartz, S. and Singer, Y.
\newblock A primal-dual perspective of online learning algorithms.
\newblock \emph{Machine Learning}, 69\penalty0 (2-3):\penalty0 115--142, 2007.

\bibitem[Tammelin et~al.(2015)Tammelin, Burch, Johanson, and
  Bowling]{Tammelin15:Solving}
Tammelin, O., Burch, N., Johanson, M., and Bowling, M.
\newblock Solving heads-up limit {T}exas hold'em.
\newblock In \emph{Proceedings of the 24th International Joint Conference on
  Artificial Intelligence (IJCAI)}, 2015.

\bibitem[{von Stengel}(1996)]{Stengel96:Efficient}
{von Stengel}, B.
\newblock Efficient computation of behavior strategies.
\newblock \emph{Games and Economic Behavior}, 14\penalty0 (2):\penalty0
  220--246, 1996.

\bibitem[von Stengel \& Forges(2008)von Stengel and
  Forges]{Stengel08:Extensive}
von Stengel, B. and Forges, F.
\newblock Extensive-form correlated equilibrium: Definition and computational
  complexity.
\newblock \emph{Mathematics of Operations Research}, 33\penalty0 (4):\penalty0
  1002--1022, 2008.

\bibitem[Wang \& Bertsekas(2013)Wang and Bertsekas]{Wang13:Incremental}
Wang, M. and Bertsekas, D.~P.
\newblock Incremental constraint projection-proximal methods for nonsmooth
  convex optimization.
\newblock \emph{SIAM J. Optim.(to appear)}, 2013.

\bibitem[Zinkevich(2003)]{Zinkevich03:Online}
Zinkevich, M.
\newblock Online convex programming and generalized infinitesimal gradient
  ascent.
\newblock In \emph{International Conference on Machine Learning (ICML)}, pp.\
  928--936, Washington, DC, USA, 2003.

\bibitem[Zinkevich et~al.(2007)Zinkevich, Bowling, Johanson, and
  Piccione]{Zinkevich07:Regret}
Zinkevich, M., Bowling, M., Johanson, M., and Piccione, C.
\newblock Regret minimization in games with incomplete information.
\newblock In \emph{Proceedings of the Annual Conference on Neural Information
  Processing Systems (NIPS)}, 2007.

\end{thebibliography}
    \bibliographystyle{custom_arxiv}

\iftrue
    \clearpage
    \appendix
    \section{Saddle-Point Formulation of EFCE}\label{app:efce saddle point}

In this section, we recall the saddle point formulation of EFCE introduced by
\citet{Farina19:Correlation}. Before we do so, we introduce the following notation:
\begin{itemize}[nolistsep,itemsep=1mm]
  \item $\seqf{i}$ denotes the sequence-form polytope of Player $i$. As mentioned in the body of the paper, this is the set of sequence-form strategies, that is
      \[
        \seqf{i} \defeq \mleft\{\vec{y} \in \bbR^{|\Sigma_i|}_+ : y[\emptyseq] = 1, \sum_{a \in A_I} y[(I, a)] = y[\sigma(I)]\quad\forall I \in \mathcal{I}_i\mright\}.
      \]
  \item Given a terminal node $z$, we denote with $\sigma_i(z)$ the last information set-action pair $(I, a)$, with $I \in \mathcal{I}_i$ and $a \in A_I$ that is encountered on the path from the root of the game tree to $z$.
  \item Given a terminal node $z$, we denote with $u_i(z)$ the utility of Player $i$ should the game end at $z$.
  \item Given a $\vec{\xi} \in \Xi$, we use the notation $\xi_1[\sigma ; z]$, where $\sigma \in \Sigma_i$ and $z$ is a terminal node, to mean $\xi[\sigma, \sigma_2(z)]$. Analogously, we use $\xi_2[\sigma ; z]$, where $\sigma\in\Sigma_2$, to mean $\xi[\sigma_1(z), \sigma]$.
  \item Given an information set $I$ for Player $i$, we denote with $Z_I$ the set of terminal nodes $z$ that include any node $v \in I$ on their path from the root of the game tree to $z$.
  \item Similarly, given a sequence $\sigma = (I, a) \in \Sigma_i$, we denote with $Z_\sigma$ the set of terminal nodes $z$ such that $\sigma(z) \succeq \sigma$. Intuitively, these are the terminal state of the game that can only be reached if Player $i$ plays action $a$ at information set $I$.
\end{itemize}

The main idea in the construction of~\citet{Farina19:Correlation} is that a correlation plan $\xi$ is an EFCE if and only if, for all player $i\in\{1, 2\}$, sequence $\sigma^* = (I^*, a^*)\in\Sigma_i$, and sequence-form strategy $\vec{y}^* \in \seqf{i}$ such that $y^*[\sigma(I^*)] = 1$ it holds that
\begin{align}\label{eq:efce incentives}
    \sum_{z\in Z_{I^*}} u_i(z) \xileaf{\sigma^*}{z} y^*(\sigma_i(z)) \le \sum_{z\in Z_{\sigma^*}} u_i(z) \xileaf{\sigma_i(z)}{z}.
\end{align}

Inequality~\eqref{eq:efce incentives} is in the form $\vec{\xi}^{\!\top}\!\! \mat{A}_{i, \sigma^*} \vec{y}_{i, \sigma^*} - \vec{b}_{i, \sigma^*}^{\!\top} \vec{\xi} \le 0$ where $\mat{A}_{i, \sigma^*}$ and $\vec{b}_{i, \sigma^*}$ are suitable matrices/vectors that only depends on the choice of player and sequence $\sigma^*$. Hence, an EFCE $\vec{\xi}$ is given by
\[
  \argmin_{\vec{\xi}\,\in\,\Xi}\max_{
    \substack{i \in \{1, 2\}\\
    \sigma^* = (I^*, a^*) \in \Sigma_i}}
        \mleft\{ \max_{\substack{\vec{y^*}\in Q_i\\y^*[\sigma(I^*)]=1}} \vec{\xi}^{\!\top}\!\! \mat{A}_{i,\sigma^*} \vec{y}^* - \vec{b}_{i,\sigma^*}^{\!\top} \vec{\xi} \mright\}.
\]

Let $\tilde{Q}_{i,I^*} \defeq \{\vec{y}^* \in \seqf{i} : y^*[\sigma(I^*)] = 1\}$. The set $\tilde{Q}_{i,I^*}$ is fundamentally equivalent to a treeplex rooted at $\sigma(I^*)$ instead of the empty sequence $\emptyseq$. As such, an any efficient regret minimizer for a treeplex (such as CFR+) can be applied to $\tilde{Q}_{i,I^*}$. In order to deal with the outer maximization, we note that the outer maximization is over a finite domain; hence, it can be converted to a unique maximization problem by introducing auxiliary nonnegative variables $\vec{\lambda} = (\lambda_{i,\sigma^*})_{i\in\{1,2\},\sigma^* \in \Sigma_i}$ such that $\vec{\lambda} \in \Delta^n$, the $n$-dimensional simplex where $n = |\{(i,\sigma^*): i\in\{1, 2\}, \sigma^* \in \Sigma_i\}|$. Therefore, an EFCE is given by
\[
\argmin_{\vec{\xi}\,\in\,\Xi}\max_{
    \substack{\vec{\lambda}\in\symp{n}\\
    {\vec{y}}_{i,\sigma^*} \in \tilde{Q}_{i,I^*}}}
        \mleft\{ \sum_{\substack{i\in\{1,2\}\\\sigma^* = (I^*, a^*)\in \Sigma_i}} \lambda_{i,\sigma^*}(\vec{\xi}^{\!\top}\!\! \mat{A}_{i,\sigma^*} {\vec{y}}_i - \vec{b}_{i,\sigma^*}^{\!\top} \vec{\xi}) \mright\}.
\]

Finally, the change of variable $\tilde{\vec{y}}_{i,\sigma^*} \defeq \lambda_{i,\sigma^*} \vec{y}_{i,\sigma^*} \in \lambda_{i,\sigma^*}\tilde{Q}_{i,I^*}$ reveals that an EFCE is the solution of the problem
\[
\argmin_{\vec{\xi}\,\in\,\Xi}\max_{
    \substack{\vec{\lambda}\in\symp{n}\\
    \tilde{\vec{y}}_{i,\sigma^*} \in \lambda_{i,\sigma^*}\tilde{Q}_{i,I^*}}}
        \mleft\{ \sum_{\substack{i\in\{1,2\}\\\sigma^* = (I^*, a^*)\in \Sigma_i}} \vec{\xi}^{\!\top}\!\! \mat{A}_{i,\sigma^*} \tilde{\vec{y}}_i - \lambda_{i,\sigma^*}\vec{b}_{i,\sigma^*}^{\!\top} \vec{\xi} \mright\},
\]
which is a bilinear saddle-point problem. An efficient regret minimizer for the domain of the maximization can be constructed by applying the convex-hull construction of~\citet{Farina19:Regret}.

As discussed by~\citet{Farina19:Correlation}, the above argument can be slightly modified to include the constraint that the EFCE $\xi$ achieve social welfare $\ge \tau$, for any given $\tau \in \bbR$.

    \section{Scaled Extension Operation}\label{app:scaled extension}

\lemscaledextension*
\begin{proof}
      Let $\cZ := \displaystyle\cX \ext^h \cY$. We break the proof into three parts:
    \begin{itemize}[leftmargin=7mm]
    \item (Non-emptiness) Since $\cX$ and $\cY$ are nonempty by hypothesis, let $\vec{x}$ and $\vec{y}$ be arbitrary points in $\cX$ and $\cY$. The element $(\vec{x}, h(\vec{x})\vec{y})$ belongs to $\cZ$ and therefore $\cZ$ is nonempty.

    \item (Compactness) We now prove that $\cZ$ is a compact set, by proving that it is bounded and closed (and applying the Heine-Borel theorem). First, we argue that $\cZ$ is bounded. Indeed, note that $h$ is affine and therefore continuous, and since $\cX$ is compact we conclude by Weierstrass' theorem that $h^* \defeq \max_{\vec{x}\in\cX} h(\vec{x})$ exists and is finite. Hence,
    $
      \cZ \subseteq \cX \times (h^* \cY) \subseteq \max\{1, h^\ast\} (\cX \times \cY)
    $
    and since both $\cX$ and $\cY$ are compact, we conclude that $\cZ$ is bounded. Now, we argue that $\cZ$ is (sequentially) closed. Indeed, let $\{\vec{z}_i\} \to \bar{\vec{z}}$ be a convergent sequence such that $\vec{z}_i \in \cZ$ for $i=1, 2,\dots$; we will prove that $\bar{\vec{z}} \in \cZ$. By definition of $\cZ$, for all $i$ it must be $\vec{z}_i = (\vec{x}_i, h(\vec{x}_i) \vec{y}_i)$ for some $(\vec{x}_i, \vec{y}_i) \in \cX \times \cY$. Since $\{\vec{z}_i\}$ converges, then the sequences $\{\vec{x}_i\}$ and $\{h(\vec{x}_i) \vec{y}_i\}$ must also converge. Let $\bar{\vec{x}} \defeq \lim \vec{x}_i$; by closedness of $\cX$, it must be $\bar{\vec{x}}\in \cX$. Furthermore, by continuity of $h$, $\lim h(\vec{x}_i) = h(\bar{\vec{x}})$. Now, using the (sequential) compactness of $\cY$, we can assume without loss of generality that $\{\vec{y}_i\}$ converges\footnote{Or else, extract a convergent subsequence.}; let $\cY \ni \bar{\vec{y}} \defeq \lim \vec{y}_i$. By the usual properties of limits,
    $
      \bar{\vec{z}} = \lim \vec{z}_i = \lim\ (\vec{x}_i, h(\vec{x}_i) \vec{y}_i) = (\bar{\vec{x}}, h(\bar{\vec{x}}) \bar{\vec{y}}),
    $
    and since $\bar{\vec{x}}\in \cX, \bar{\vec{y}}\in\cY$ we have $\bar{\vec{z}} \in \cZ$ and $\cZ$ is sequentially closed.

    \item (Convexity) Since $\cZ$ was just proven to be compact, it is in particular closed and hence it will be enough to prove midpoint convexity to conclude convexity of $\cZ$. To this end, let $\vec{z}, \vec{z}'$ be any two points in $\cZ$. By definition of $\cZ$, there must exist $\vec{x},\vec{x}' \in \cX$ and $\vec{y},\vec{y}'\in \cY$ such that $\vec{z}=(\vec{x}, h(\vec{x})\vec{y})$ and $\vec{z}'=(\vec{x}', h(\vec{x}')\vec{y}')$. If $h(\vec{x}) = h(\vec{x}') = 0$, then the affinity of $h$ implies $h(\vec{x}/2 + \vec{x}'/2) = 0$ and therefore
    \begin{align*}
      \frac{\vec{z} + \vec{z}'}{2} = \mleft(\frac{\vec{x} + \vec{x}'}{2}, \vec{0}\mright) = \mleft(\frac{\vec{x} + \vec{x}'}{2},\ h\!\mleft(\frac{\vec{x} + \vec{x}'}{2}\mright) \vec{y}\mright) \in \cZ.
    \end{align*}On the other hand, if at least one between $h(\vec{x})$ and $h(\vec{x}')$ is strictly positive, then
    \begin{align*}
      \frac{\vec{z} + \vec{z}'}{2} &= \mleft(\frac{\vec{x} + \vec{x}'}{2}, \frac{h(\vec{x})\vec{y} + h(\vec{x}') \vec{y}'}{2}\mright) \\
        &= \mleft(\frac{\vec{x} + \vec{x}'}{2}, \frac{h(\vec{x}) + h(\vec{x}')}{2}\mleft[\frac{h(\vec{x})}{h(\vec{x}) + h(\vec{x}')}\vec{y} + \frac{h(\vec{x}')}{h(\vec{x}) + h(\vec{x}')}\vec{y}'\mright]\mright)\\
        &= \mleft(\frac{\vec{x} + \vec{x}'}{2},\ h\!\mleft(\frac{\vec{x}+\vec{x}'}{2}\mright) \mleft[\frac{h(\vec{x})}{h(\vec{x}) + h(\vec{x}')}\vec{y} + \frac{h(\vec{x}')}{h(\vec{x}) + h(\vec{x}')}\vec{y}'\mright]\mright),
    \end{align*}
    where the last equality follows from the fact that $h$ is an affine function. Since $h$ is non-negative, the convex combination in the square brackets belongs to $\cY$ and therefore $\vec{z}/2 + \vec{z}'/2 \in \cZ$. \qedhere
\end{itemize}
\end{proof}

\section{Regret Circuit for the Scaled Extension Operation}\label{app:regret circuit}

The expression for the cumulative regret of a generic sequence of decisions $(\vec{x}^t, h(\vec{x}^t) \vec{y}^t) \in \cZ$ for $t = 1, \dots, T$ is
\begin{align*}
  R^T_{\cZ} &= \sum_{t=1}^T \Big(\langle \vec{\ell}_{{x}}^t, \vec{x}^t \rangle + h(\vec{x}^t) \langle \vec{\ell}^t_y, \vec{y}^t\rangle\Big) - \min_{\substack{\hat{\vec{x}} \in \cX\\\hat{\vec{y}}\in \cY}} \mleft\{\sum_{t=1}^T\Big( \langle \vec{\ell}_{{x}}^t, \hat{\vec{x}} \rangle + h(\hat{\vec{x}}) \langle \vec{\ell}^t_y, \hat{\vec{y}}\rangle\Big) \mright\} \\
   &= \sum_{t=1}^T \Big(\langle \vec{\ell}_{{x}}^t, \vec{x}^t \rangle + h(\vec{x}^t) \langle \vec{\ell}^t_y, \vec{y}^t\rangle\Big) - \min_{\substack{\hat{\vec{x}} \in \cX\\\hat{\vec{y}}\in\cY}} \mleft\{\mleft(\sum_{t=1}^T \langle \vec{\ell}_{{x}}^t, \hat{\vec{x}} \rangle\mright) + h(\hat{\vec{x}}) \sum_{t=1}^T\langle \vec{\ell}^t_y, \hat{\vec{y}}\rangle \mright\}.
\end{align*}
Since $h$ is nonnegative, we can separately minimize the last sum $\sum_{t=1}^T \langle \vec{\ell}_{{y}}^t, \hat{\vec{y}}\rangle$, and obtain
\begin{align}
  R^T_{\cZ} &= \sum_{t=1}^T \!\Big(\langle \vec{\ell}_{{x}}^t, \vec{x}^t \rangle + h(\vec{x}^t) \langle \vec{\ell}^t_y, \vec{y}^t\rangle\Big) - \min_{\hat{\vec{x}} \in \cX} \!\mleft\{\!\mleft(\sum_{t=1}^T \langle \vec{\ell}_{{x}}^t, \hat{\vec{x}} \rangle\mright) + h(\hat{\vec{x}}) \min_{\hat{\vec{y}}\in \cY} \mleft\{ \sum_{t=1}^T\langle \vec{\ell}^t_y, \hat{\vec{y}}\rangle \mright\}\!\mright\}\!.\label{eq:regret 1}
\end{align}

This suggests the following natural idea: we let $\vec{y}^t$ be chosen by a regret minimizer for $\cY$ that observes losses $\vec{\ell}^t$ at each time $t$. This way, by definition of cumulative regret we have
\[
  \min_{\hat{\vec{y}} \in \cY} \mleft\{\sum_{t=1}^T\langle \vec{\ell}^t_y, \hat{\vec{y}}\rangle  \mright\} = -R^T_\cY + \sum_{t=1}^T\langle \vec{\ell}^t_y, \vec{y}^t\rangle,
\]
and substituting into~\eqref{eq:regret 1} we obtain
\begin{align*}
  R^T_{\cZ} &= \sum_{t=1}^T \Big(\langle \vec{\ell}_{{x}}^t, \vec{x}^t \rangle + h(\vec{x}^t) \langle \vec{\ell}^t_y, \vec{y}^t\rangle\Big) - \min_{\hat{\vec{x}} \in \cX} \mleft\{\sum_{t=1}^T \Big( \langle \vec{\ell}_{{x}}^t, \hat{\vec{x}} \rangle + h(\hat{\vec{x}}) \langle\vec{\ell}_{{y}}^t, \vec{y}^t\rangle \Big) - h(\hat{\vec{x}})R_\cY^T \mright\} \\
   &\le \sum_{t=1}^T \Big(\langle \vec{\ell}_{{x}}^t, \vec{x}^t \rangle + h(\vec{x}^t) \langle \vec{\ell}^t_y, \vec{y}^t\rangle\Big) \!-\! \min_{\hat{\vec{x}} \in \cX} \mleft\{\sum_{t=1}^T \Big( \langle \vec{\ell}_{{x}}^t, \hat{\vec{x}} \rangle + h(\hat{\vec{x}}) \langle\vec{\ell}_{{y}}^t, \vec{y}^t\rangle \Big)\!\mright\} \!+ \max_{\hat{\vec{x}}\in\cX}\! \Big\{ h(\hat{\vec{x}})R_\cY^T \Big\}.
\end{align*}

Since $h$ is affine, it can be expressed as $h: \vec{x} \mapsto \langle \vec{a}, \vec{x}\rangle + {b}$  for some vectors $\vec{a}, {b}$. Hence,
\begin{align*}
  R^T_{\cZ} &\le \sum_{t=1}^T \Big\langle \vec{\ell}_{{x}}^t + \langle \vec{\ell}_{{y}}^t, \vec{y}^t\rangle \vec{a}, \vec{x}^t\Big\rangle - \min_{\hat{\vec{x}} \in \cX} \mleft\{\sum_{t=1}^T  \Big\langle \vec{\ell}_{{x}}^t + \langle \vec{\ell}_{{y}}^t, \vec{y}^t\rangle \vec{a}, \hat{\vec{x}} \Big\rangle \mright\} + \max_{\hat{\vec{x}}\in\cX} \Big\{ h(\hat{\vec{x}})R_\cY^T \Big\}.
\end{align*}
The first two terms in the difference in the right-hand side correspond to the cumulative regret $R_\cX^T$ of a regret minimizer for $\cX$ that observes \begin{equation}\label{eq:tilde ell}
  \tilde{\vec{\ell}}_{{x}}^t \defeq \vec{\ell}_{{x}}^t + \langle \vec{\ell}_{{y}}^t, \vec{y}^t\rangle \vec{a}
\end{equation}
at all times $t$. Hence, by denoting $h^* \defeq \max_{\hat{\vec{x}}\in\cX} h(\hat{\vec{x}})$,\footnote{Since $h$ is affine and $\cX$ is compact, $h^\ast$ exists and is finite by Weierstrass' theorem.} we obtain
$
  R_\cZ^T \le R_\cX^T + h^\ast R_\cY^T.
$
In other words, as long as the regret minimizers for $\cX$ and $\cY$ are Hannan consistent, so is the regret minimizer defined by \cref{fig:circuit}.

\begin{figure}
\centering\includegraphics[scale=.8]{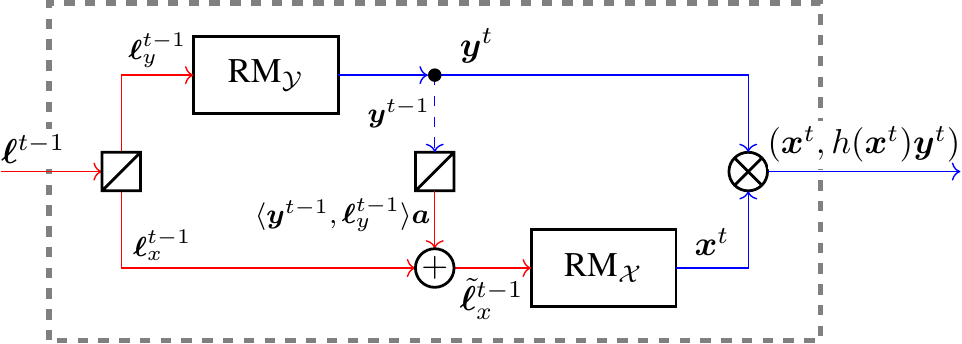}
\caption{Regret minimizer for the set $\displaystyle\cZ = \cX \ext^h \cY$. The affine function $h$ is expressed as $h : \cX \ni \vec{x} \mapsto \langle \vec{a}, \vec{x}\rangle + {b}$. The loss $\tilde{\vec{\ell}}_x^{t-1}$ is defined in~\eqref{eq:tilde ell}.}
\label{fig:circuit}
\end{figure}

\section{Proof of Proposition~\ref{prop:unfavorable}}\label{app:unfavorable}
\propunfavorable*
\begin{proof}
  For contradiction, assume that $I_1\conn I_2$ and $I_1' \conn I_2'$. By \cref{def:connected infosets}, there exists two pairs of connected nodes $(u, v) \in I_1\times I_2$ and $(u', v') \in I_1'\times I_2'$. Let $w$ be the lowest common ancestor of $u$ and $u'$. It cannot be that $w=u$ or $w=u'$, or it would not be true that $\sigma(I_1) = \sigma(I_1')=\sigma_1$. Also, $w$ cannot be a node for Player 1, or the game would not be perfect-recall. Hence, since the game does not have chance moves, $w$ must belong to Player 2. Since $w$ belongs to Player 2 and $\sigma(I_2)=\sigma(I_2')=\sigma_2$, it is not possible that $v, v'$ be descendants of $w$ (or the game would not be perfect-recall). Hence, it must be that both $v$ and $v'$ belong to the path from $w$ to the root of the game tree. But then one between $I_2$ and $I_2'$ must precede the other one, contradicting the fact that $\sigma(I_2) = \sigma(I_2') = \sigma_2$.
\end{proof} 
    \section{Polynomial-Time Algorithm that Decomposes $\Xi$ using Scaled Extensions}\label{app:algo}

    We propose pseudocode for the algorithm presented in \cref{sec:unroll algo} in \cref{algo:main}. We use the following conventions:
\begin{itemize}[leftmargin=*,nolistsep,itemsep=0mm]
  \item Given a player $i\in\{1,2\}$, we let $-i$ denote the opponent.
  \item We use the symbol $\sqcup$ to denote disjoint union.
  \item Given two infosets $I,I' \in \mathcal{I}_i$, we write $I \preceq I'$ if $\sigma(I') \succeq \sigma(I)$. We say that we iterate over a set $\mathcal{I} \subseteq \mathcal{I}_i$ \emph{in top-down order} if, given any two $I,I'\in\mathcal{I}$ such that $I\preceq I'$, $I$ appears before $I'$ in the iteration.
  \item We use the observation that for all $I\in\mathcal{I}_1$ and $\sigma_2\in\Sigma_2$, $I \rele \sigma_2$ if and only if $(I, a) \rele \sigma_2 \ \forall a\in A_I$. (A symmetric statement holds for $I \in \mathcal{I}_2$ and $\sigma_1\in\Sigma_1$.)
\end{itemize}

    \begin{algorithm}[ht]
        \caption{Pseudocode for the algorithm of \cref{sec:unroll algo}.}
        \label{algo:main}
        \begin{minipage}{0.99\linewidth}
            \begin{algorithmic}[1]
              \Function{Decompose}{$(\sigma_1,\sigma_2), \mathcal{S}, \mathcal{D}$}
              \State Let $i^* \in \{1,2\}$ be a critical player for $(\sigma_1,\sigma_2) \in \Sigma_1\times\Sigma_2$
              \State $\mathcal{I}^* \gets \{I \in \mathcal{I}_{i^*} : I \text{ is critical for } \sigma_{-i^*}\} \subseteq \mathcal{I}_{i^*}$
                \Statex \Comment\textcolor{gray}{By definition of critical player, $\mathcal{I}^*$ is either a singleton or an empty set}
              \For{\textbf{each} $I \in \mathcal{I}_{i^*}$ such that $\sigma(I) = \sigma_{i^*}$ and $\sigma_{-i}^* \rele I$}
                  \If{$i^* = 1$}
                    \State $\mathcal{S}' \gets \{((I,a), \sigma_2): a \in A_{I}\}$
                  \Else
                    \State $\mathcal{S}' \gets \{(\sigma_1, (I,a)): a \in A_{I}\}$
                  \EndIf
                  \State $\mathcal{D} \gets \mathcal{D} \ext^h \symp{|A_{I}|}$, where $h : \mathcal{D} \ni \vec{\xi} \mapsto \xi[\sigma_1,\sigma_2]$.\label{line:update D 1}
                    \Statex \Comment\textcolor{gray}{The $|A_{I}|$ new entries correspond to indices $\mathcal{S}'$}
                  \State $\mathcal{S} \gets \mathcal{S} \sqcup \mathcal{S}'$\label{line:update S 1}
                  \For{\textbf{each} $a \in A_{I}$}
                    \If{$i^* = 1$}
                      \State $(\mathcal{S}, \mathcal{D}) \gets \textsc{Decompose}(((I,a),\sigma_2), \mathcal{S}, \mathcal{D})$ \Comment\textcolor{gray}{Recursive step}\label{line:update S 2}\label{line:update D 2}
                    \Else
                      \State $(\mathcal{S}, \mathcal{D}) \gets \textsc{Decompose}((\sigma_1,(I, a)), \mathcal{S}, \mathcal{D})$ \Comment\textcolor{gray}{Recursive step}\label{line:update S 3}\label{line:update D 3}
                    \EndIf
                  \EndFor
            \EndFor
                  \For{\textbf{each} $J \in \mathcal{I}_{-i^*}$ such that $\sigma(J) \succeq \sigma_{-i^*}$ and $\sigma_{i^*} \rele J$, in \emph{top-down order},}
                    \If{$\mathcal{I}^* = \{I^*\}$ for some $I^* \in \mathcal{I}_{i^*}$ \textbf{and} $I^* \conn J$}\label{line:if}
                        \For{\textbf{each} $a \in A_J$}
                            \If{$i^* = 1$}
                                \State $\mathcal{S}' \gets \{(\sigma_1, (J,a))\}$
                          \State $\mathcal{D} \gets \mathcal{D} \ext^h \{1\}$, where $h : \mathcal{D} \ni \vec{\xi} \mapsto \sum_{a^* \in A_{I^*}}\xi[(I^*, a^*),(J,a)]$\label{line:update D 4}
                            \Statex\Comment\textcolor{gray}{The new entry corresponds to index $\mathcal{S}'$}
                            \Else
                                \State $\mathcal{S}' \gets \{((J,a), \sigma_2)\}$
                              \State $\mathcal{D} \gets \mathcal{D} \ext^h \{1\}$, where $h : \mathcal{D} \ni \vec{\xi} \mapsto \sum_{a^* \in A_{I^*}}\xi[(J,a), (I^*, a^*)]$\label{line:update D 5}
                                \Statex \Comment\textcolor{gray}{The new entry corresponds to index $\mathcal{S}'$}
                            \EndIf
                            \State $\mathcal{S} \gets \mathcal{S} \sqcup \mathcal{S}'$\label{line:update S 4}
                        \EndFor
                    \Else
                      \If{$i^* = 1$}
                        \State $\mathcal{S}' \gets \{(\sigma_1, (J, a)): a \in A_J\}$
                          \State $\mathcal{D} \gets \mathcal{D} \ext^h \symp{|A_J|}$, where $h : \mathcal{D} \ni \vec{\xi} \mapsto \xi[\sigma_1,\sigma(J)]$
                            \Statex \Comment\textcolor{gray}{The $|A_J|$ new entries\label{line:update D 6} correspond to indices $\mathcal{S}'$}
                      \Else
                        \State $\mathcal{S}' \gets \{((J,a), \sigma_2): a \in A_J\}$
                          \State $\mathcal{D} \gets \mathcal{D} \ext^h \symp{|A_J|}$, where $h : \mathcal{D} \ni \vec{\xi} \mapsto \xi[\sigma(J),\sigma_2]$
                            \Statex \Comment\textcolor{gray}{The $|A_J|$ new entries\label{line:update D 7} correspond to indices $\mathcal{S}'$}
                      \EndIf
                      \State $\mathcal{S} \gets \mathcal{S} \sqcup \mathcal{S'}$\label{line:update S 5}
                    \EndIf
                  \EndFor
              \State \textbf{return} $(\mathcal{S}, \mathcal{D})$
              \EndFunction
            \end{algorithmic}
        \end{minipage}
    \end{algorithm}

As stated in \cref{sec:unroll algo}, the outermost call to \textsc{Decompose} is $\textsc{Decompose}((\emptyseq,\emptyseq), \mathcal{S} = \{(\emptyseq,\emptyseq)\}, \mathcal{D} = \{ (1) \})$, which corresponds to the starting situation in which only the entry $\xi[\emptyseq,\emptyseq]$ has been filled in (with the value 1 as per \cref{def:xi}).
The correctness of the algorithm relies fundamentally on the following inductive contract:
\begin{lemma}[Inductive contract]\label{lem:contract}
  At the beginning of each call to $\textsc{Decompose}((\sigma_1,\sigma_2), \mathcal{S}, \mathcal{D})$,
\begin{enumerate}[leftmargin=1.5cm,nolistsep,itemsep=1mm,label=(Pre\arabic*)]
  \item $\mathcal{S}$ contains only relevant sequence pairs.
  \item $\mathcal{D}$ consists of vectors indexed by exactly the indices in $\mathcal{S}$.
  \item $\mathcal{S}$ does not contain any relevant sequence pairs which are descendants of $(\sigma_1,\sigma_2)$, with the only exception of $(\sigma_1,\sigma_2)$ itself. In formulas, \[\mathcal{S} \cap \{(\sigma'_1,\sigma'_2) \in \Sigma_1\times\Sigma_2 : \sigma'_1 \succeq \sigma_1, \sigma'_2 \succeq \sigma_2\} = \{(\sigma_1, \sigma_2)\}.\]
\end{enumerate}

  At the end of the call, the return value $(\mathcal{S}', \mathcal{D}')$ is such that
\begin{enumerate}[leftmargin=1.5cm,nolistsep,itemsep=1mm,label=(Post\arabic*)]
\item $\mathcal{S}'$ contains only relevant sequence pairs.
\item $\mathcal{D}'$ consists of vectors $\xi$ indexed by exactly the indices in $\mathcal{S}'$.
\item The call has filled in exactly all relevant sequence pair indices that are descendants of $(\sigma_1,\sigma_2)$ (except for $(\sigma_1,\sigma_2)$ itself, which was already filled in). In formulas, \[\mathcal{S}' = \mathcal{S} \sqcup \{(\sigma'_1,\sigma'_2) \in \Sigma_1\times\Sigma_2 : \sigma'_1 \succeq \sigma_1, \sigma'_2 \succeq \sigma_2, (\sigma'_1,\sigma'_2) \neq (\sigma_1,\sigma_2), \sigma'_1 \rele \sigma'_2\}.\]
\item $\mathcal{D}'$ satisfies a subset of constraints of \cref{def:xi}:
\end{enumerate}
    \begin{equation*}
  \mathcal{D}' \subseteq \left\{\vec{\xi}\ge\vec{0}: \!\!\begin{array}{ll}
    \circled{1}\ \!\displaystyle\sum_{a \in A_I}~\!\xi[(I, a),\hspace{.1mm} \sigma'_2] = \xi[\sigma(I),\hspace{.4mm} \sigma'_2] & \forall \sigma'_2 \succeq \sigma_2, I \in \mathcal{I}_1 \ \hspace{.0mm} \text{ s.t. } \sigma'_2 \rele I, \sigma(I) \succeq \sigma_1\\[8mm]
    \circled{2}\ \!\displaystyle\sum_{a \in A_J} \xi[\sigma'_1, (J, a)] = \xi[\sigma'_1, \sigma(J)] & \forall \sigma'_1\succeq\sigma_1, J \in \mathcal{I}_2 \text{ s.t. } \sigma'_1 \rele J, \sigma(J) \succeq \sigma_2
  \end{array}\!\!\!\right\}\!.
  \end{equation*}
\end{lemma}
\begin{proof}~
  \begin{itemize}
    \item \emph{(Pre1) and (Post1)}. We prove that if $\mathcal{S}$ only contains relevant sequence pairs at the beginning of the call, then the returned $\mathcal{S}'$ only contains relevant sequence pairs. The proof is by induction on the call tree, where the base case is any call where $(\sigma_1,\sigma_2)$ is such that $\{\sigma'_1 \succeq \sigma_1 : \sigma'_1 \rele \sigma_2\} = \emptyset$ and $\{\sigma'_2 \succeq \sigma_2 : \sigma_1 \rele \sigma'_2\} = \emptyset$, for which no further call to \textsc{Decompose} is performed. The only updates to $\mathcal{S}$ happen at Lines~\ref{line:update S 1},~\ref{line:update S 2},~\ref{line:update S 3},~\ref{line:update S 4} and~\ref{line:update S 5}:
        \begin{itemize}
          \item Line~\ref{line:update S 1}. Since $\sigma_{-i^*} \rele I$, by definition $\mathcal{S}'$ contains relevant sequence pairs.
          \item Lines~\ref{line:update S 2} and~\ref{line:update S 3}. Follows by the inductive step.
          \item Lines~\ref{line:update S 4} and~\ref{line:update S 5}. Since $\sigma_{i^*} \rele J$, by definition $\mathcal{S}'$ contains relevant sequence pairs.
        \end{itemize}
    \item \emph{(Pre 2) and (Post 2)}. We prove that if $\mathcal{D}$ contains vectors indexed by exactly the indices in $\mathcal{S}$,then the returned $(\mathcal{S}', \mathcal{D}')$ is such that $\mathcal{D}'$ contains vectors indexed by exactly the indices in $\mathcal{S}$. Again, the proof is by induction on the call tree, where the base case is any call where $(\sigma_1,\sigma_2)$ is such that $\{\sigma'_1 \succ \sigma_1 : \sigma'_1 \rele \sigma_2\} = \emptyset$ and $\{\sigma'_2 \succ \sigma_2 : \sigma_1 \rele \sigma'_2\} = \emptyset$, for which no further call to \textsc{Decompose} is performed. The only updates to $\mathcal{D}$ happen at Lines~\ref{line:update D 1},~\ref{line:update D 2},~\ref{line:update D 3},~\ref{line:update D 4},~\ref{line:update D 5},~\ref{line:update D 6} and~\ref{line:update D 7}. All cases are trivial.
    \item \emph{(Pre 3)}. By induction on the call order. The base cases (initial call) follows since $\mathcal{S} = \{(\emptyseq,\emptyseq)\}$. Hence, it is enough to prove that Line~\ref{line:update S 1} maintains the property. Fix a call with parameter $(\sigma_1,\sigma_2)$. Note that because of the inductive hypothesis, we have
        \begin{equation}\label{eq:indu1}
        \mathcal{S} \cap \{(\sigma'_1,\sigma'_2) \in \Sigma_1\times\Sigma_2 : \sigma'_1 \succeq \sigma_1, \sigma'_2 \succeq \sigma_2\} = \{(\sigma_1, \sigma_2)\}.
        \end{equation}
        Assume without loss of generality that $i^* = 1$ (the other case is symmetric). Then, at each iteration of Line~\ref{line:update S 2}, from \eqref{eq:indu1} we have
        \begin{equation}\label{eq:indu2}
            \mathcal{S} \cap \{(\sigma'_1,\sigma'_2) \in \Sigma_1\times\Sigma_2 : \sigma'_1 \succeq (I,a), \sigma'_2 \succeq \sigma_2\} = \emptyset.
        \end{equation}

Since from Line~\ref{line:update S 1} $\mathcal{S}$ was updated by taking the union with
        \[
          \mathcal{S'} = \{((I,a), \sigma_2): a \in A_I\},
        \]
        then for all $a \in A_I$
        \begin{align*}
          &(\mathcal{S} \cup \mathcal{S'}) \cap \{(\sigma'_1,\sigma'_2) \in \Sigma_1\times\Sigma_2 : \sigma'_1 \succeq (I, a), \sigma'_2 \succeq \sigma_2 \} \\
  &\hspace{2cm}= \mathcal{S'} \cap \{(\sigma'_1,\sigma'_2) \in \Sigma_1\times\Sigma_2 : \sigma'_1 \succeq (I, a), \sigma'_2 \succeq \sigma_2\} = \{((I,a), \sigma_2)\}.
        \end{align*}
        We also note that (Pre3) implies that the disjoint unions of Lines~\ref{line:update S 1},~\ref{line:update S 4} and~\ref{line:update S 5} are effectively disjoint unions.

    \item \emph{(Post3)}. By induction on the call tree. The base case is any call where $(\sigma_1,\sigma_2)$ is such that $\{\sigma'_1 \succ \sigma_1 : \sigma'_1 \rele \sigma_2\} = \emptyset$ and $\{\sigma'_2 \succ \sigma_2 : \sigma_1 \rele \sigma'_2\} = \emptyset$, for which no further call to \textsc{Decompose} is performed. In that case, we see that the algorithm does not fill in any new index, and therefore the claim holds. Consider now a call with relevant sequence pair $(\sigma_1,\sigma_2)$. Given (Pre3), (Pre2) and (Post2), it suffices to see what indices are added to $\mathcal{S}$ at Lines~\ref{line:update S 1},~\ref{line:update S 2},~\ref{line:update S 3},~\ref{line:update S 4} and~\ref{line:update S 5}. Assume without loss of generality that $i^*=1$ (the other case is symmetric). Using the inductive hypothesis, we see that for each $I \in \mathcal{I}_{i^*}$ such that $\sigma(I) = \sigma_{i^*}$ and $\sigma_{-i^*} \rele I$, Lines~\ref{line:update S 1},~\ref{line:update S 2} and~\ref{line:update S 3} add exactly indices \[\{(\sigma'_1,\sigma'_2) : \sigma'_1 \rele I, \sigma'_1 \succ \sigma_1, \sigma'_2 \succeq \sigma_2, \sigma'_1 \rele \sigma'_2\}.\] Taking the union over all $I \in \{ I \in \mathcal{I}_{i^*} : \sigma(I) = \sigma_{i^*}, \sigma_{-i^*} \rele I \}$, we see that the only indices missing from our target $\{(\sigma'_1,\sigma'_2) \in \Sigma_1\times\Sigma_2 : \sigma'_1 \succeq \sigma_1, \sigma'_2 \succeq \sigma_2, (\sigma'_1,\sigma'_2) \neq (\sigma_1,\sigma_2), \sigma'_1 \rele \sigma'_2\}$ are those from the set $\{(\sigma_1,\sigma'_2) \in \Sigma_1\times\Sigma_2 : \sigma'_2 \succ \sigma_2, \sigma_1 \rele \sigma'_2\}$. These are exactly the indices that correspond to all the sequences $(J, a): a \in A_J$, for all $J \in \mathcal{I}_{2},\sigma(J) \succeq \sigma_{2}, \sigma_1 \rele J$. These indices are filled in on \cref{line:update D 4,line:update D 5,line:update D 6,line:update D 7} and added to $\mathcal{S}$ on \cref{line:update S 4,line:update S 5}.

        The analysis so far also reveals that all the functions $h$ used in Lines~\ref{line:update D 1},~\ref{line:update D 2},~\ref{line:update D 3},~\ref{line:update D 4},~\ref{line:update D 5},~\ref{line:update D 6} and~\ref{line:update D 7} are well-defined.

    \item \emph{(Post4)} By induction on the call tree. The base case is any call where $(\sigma_1,\sigma_2)$ is such that $\{\sigma'_1 \succ \sigma_1 : \sigma'_1 \rele \sigma_2\} = \emptyset$ and $\{\sigma'_2 \succ \sigma_2 : \sigma_1 \rele \sigma'_2\} = \emptyset$, for which no further call to \textsc{Decompose} is performed. In that case, we see that the algorithm does not fill in any new index and that the set of constraints is empty, and therefore the claim holds. Assume without loss of generality that $i^*=1$ (the other case is symmetric). \cref{line:update D 1} guarantees that
        \begin{equation*}
            \bullet\ \ \!\displaystyle\sum_{a \in A_I}~\!\xi[(I, a),\hspace{.1mm} \sigma_2] = \xi[\sigma_1,\hspace{.4mm} \sigma_2] \hspace{2cm} \forall I \in \mathcal{I}_1 \text{ s.t. } \sigma_2 \rele I, \sigma(I) = \sigma_1.
        \end{equation*}
        By using the inductive hypothesis, \cref{line:update D 2} (and \cref{line:update D 3} in the case $i^* = 2$) guarantees that all the constraints
        \[
        \begin{array}{ll}
            \bullet\ \ \!\displaystyle\sum_{a \in A_I}~\!\xi[(I, a),\hspace{.1mm} \sigma'_2] = \xi[\sigma(I),\hspace{.4mm} \sigma'_2] & \forall \sigma'_2 \succeq \sigma_2, I \in \mathcal{I}_1 \ \hspace{.0mm} \text{ s.t. } \sigma'_2 \rele I, \sigma(I) \succ \sigma_1\\[5mm]
    \bullet\ \ \!\displaystyle\sum_{a \in A_J} \xi[\sigma'_1, (J, a)] = \xi[\sigma'_1, \sigma(J)] & \forall \sigma'_1\succ\sigma_1, J \in \mathcal{I}_2 \text{ s.t. } \sigma'_1 \rele J, \sigma(J) \succeq \sigma_2
        \end{array}
        \]
        are satisfied (note the strict $\succ$). Hence, it is enough to prove that Lines~\ref{line:update D 4} and~\ref{line:update D 6} fill in all indices $\{(\sigma_1, \sigma'_2): \sigma'_2 \succ \sigma_2, \sigma_1 \rele \sigma'_2\}$ in such a way that all constraints
\[
\begin{array}{ll}
        \circled{A}\ \ \!\displaystyle\sum_{a \in A_J} \xi[\sigma_1, (J, a)] = \xi[\sigma_1, \sigma(J)] & \forall J \in \mathcal{I}_2 \text{ s.t. } \sigma_1 \rele J, \sigma(J) \succeq \sigma_2\\[5mm]
            \circled{B}\ \ \!\displaystyle\sum_{a \in A_I}~\!\xi[(I, a),\hspace{.1mm} \sigma'_2] = \xi[\sigma_1,\hspace{.4mm} \sigma'_2] & \forall \sigma'_2 \succ \sigma_2, I \in \mathcal{I}_1 \ \hspace{.0mm} \text{ s.t. } \sigma'_2 \rele I, \sigma(I) = \sigma_1
        \end{array}
\]
    hold.
    To this end, note that Line~\ref{line:update D 4} (and~\ref{line:update D 5} in the case $i^* = 2$) guarantees constraints \circled{B}, while Line~\ref{line:update D 6} (and~\ref{line:update D 7} in the case $i^* = 2$) guarantees constraints \circled{A} for all $J$ such that the condition of the \textbf{if} statement on Line~\ref{line:if} is \emph{not} met. Hence, it is enough to show that for all $J$ such that the condition of the \textbf{if} statement on Line~\ref{line:if} is met, constraint \circled{A} holds. This is easy to show: for any such $J$, all the entries of $\xi[\sigma_1, (J,a)]$ ($a\in A_J$) were filled in Line~\ref{line:update D 4}, and thus we have
\begin{align*}
  \sum_{a \in A_J} \xi[\sigma_1, (J,a)] &= \sum_{a\in A_J} \sum_{b \in A_{I^*}} \xi[(I^*, b), (J,a)] \\
  &= \sum_{b\in A_{I^*}}\sum_{a\in A_J} \xi[(I^*, b), (J,a)] \\
  &= \sum_{b\in A_{I^*}} \xi[(I^*, b), \sigma(J)]\\
  &= \xi[\sigma(I^*), \sigma(J)] = \xi[\sigma_1, \sigma(J)],
\end{align*}
  where the second-to-last equality comes from the observation that $I^* \conn J$ implies $\sigma(J) \rele I^*$. This concludes the proof. \qedhere
  \end{itemize}
\end{proof}

Conditions (Post2), (Post3) and (Post4) in \cref{lem:contract} together imply that when the algorithm terminates (with return value $(\mathcal{S}',\mathcal{D}')$), $\mathcal{S}'$ is the set of relevant sequences in the game, and that $\mathcal{D}' \subseteq \Xi$. Hence, it is enough to show that $\mathcal{D}' \supseteq \Xi$ to conclude the proof of correctness of the algorithm:

\begin{lemma}
  At termination, the set $\mathcal{D}$ returned by the algorithm is such that $\Xi \subseteq \mathcal{D}$.
\end{lemma}
  We need to prove that any $\vec{\xi}$ that satisfies the constraints of~\cref{def:xi} appears in $\mathcal{D}$. This is rather straightforward (we focus on the updates to $\mathcal{D}$ relevant to $i^* = 1$ only---the other case is symmetric):
\begin{itemize}
  \item Line~\ref{line:update D 1}. According to the constraints in~\ref{def:xi}, it must be
      $
        \sum_{a \in A_I} \xi[(I,a),\sigma_2] = \xi[\sigma_1, \sigma_2].
      $
      Hence, the vector of entries $(\xi[(I,a),\sigma_2])_{a\in A_I}$ indeed belongs to $\xi[\sigma_1,\sigma_2] \Delta^{|A_I|}$.
  \item Line~\ref{line:update D 4}. According to the constraints in~\ref{def:xi}, and since $I^* \conn J$, it must be for all $J$ in the iteration and for all $a\in A_J$:
      $
        \sum_{a^* \in A_{I^*}} \xi[(I^*,a^*),(J,a)] = \xi[\sigma(I^*), (J,a)] = \xi[\sigma_1, (J,a)].
      $
      Hence, $\{\xi[\sigma_1, (J,a)]\} \in \big(\sum_{a^* \in A_{I^*}} \xi[(I^*,a^*),(J,a)]\big) \{ 1 \}$.
  \item Line~\ref{line:update D 6}. According to the constraints in~\ref{def:xi}, for all $J$ in the iteration it must be
      $
        \sum_{a \in A_J} \xi[\sigma_1, (J,a)] = \xi[\sigma_1, \sigma(J)].
      $
      Hence, the vector of entries $(\xi[\sigma_1, (J,a)])_{a\in A_J}$ indeed belongs to $\xi[\sigma_1,\sigma(J)] \Delta^{|A_J|}$. \qedhere
\end{itemize} 
    \section{Regret Minimization Algorithm for $\Xi$}
\label{app:full rm algo}

In this section, we give the pseudocode of our regret minimization algorithm for the space $\Xi$ of correlated strategies in a two-player extensive-form game with no chance moves.
We use the following notation:
\begin{itemize}
  \item $\textsc{FillSimplex}((\sigma_1, \sigma_2) \to I)$, where $I \in \mathcal{I}_1$ (respectively, $I \in \mathcal{I}_2$) is such that $\sigma(I) = \sigma_1$ (respectively, $\sigma(I) = \sigma_2$) corresponds to filling in all entries $\{\xi[(I, a), \sigma_2]: a \in A_{I}\}$ (respectively, $\{\xi[\sigma_1, (I, a)]: a \in A_I\}$) given $\xi[\sigma_1, \sigma_2]$. This can be expressed via the scaled extension $\mathcal{D} \ext^h \Delta^{|A_{I}|}$, where $h: \vec{\xi} \mapsto \xi[\sigma_1, \sigma_2]$ as already discussed in the body of the paper.
  \item $\textsc{SumSimplex}((\sigma_1, \sigma_2) \gets I)$, where $I \in \mathcal{I}_1$ (respectively, $I \in \mathcal{I}_2$) is such that $\sigma(I) = \sigma_1$ (respectively, $\sigma(I) = \sigma_2$) corresponds to filling in $\xi[\sigma_1, \sigma_2]$ by assigning it to the sum $\sum_{a\in A_I} \xi[(I, a), \sigma_2]$ (respectively, $\sum_{a\in A_I} \xi[\sigma_1, (I,a)]$).
\end{itemize}

\begin{algorithm}[ht]
  \caption{Regret minimizer for the correlated strategy polytope $\Xi$ in a two-player game with no chance moves}\label{algo:full rm}
  \begin{algorithmic}[1]
    \Function{Recommend}{~\!}
        \State $\vec{\xi} \gets \vec{0} \in \bbR^{|\Xi|}$\Comment{\textcolor{gray}{$|\Xi|$ is the number of relevant sequence pairs in the game}}
        \For{$i=1,\dots,n$}
            \If{$\textsf{op}_i = \textsc{FillSimplex}((\sigma_1, \sigma_2) \to I)$}
                \If{$I \in \mathcal{I}_1$} \Comment{\textcolor{gray}{$I$ belongs to Player 1}}
                    \State $(\vec{\xi}[(I, a), \sigma_2])_{a \in A_I} \gets \xi[\sigma_1, \sigma_2] \cdot \text{RM}_i.\textsc{Recommend}()$
                \Else
                    \State $(\vec{\xi}[\sigma_1, (I, a)])_{a \in A_I} \gets \xi[\sigma_1, \sigma_2] \cdot \text{RM}_i.\textsc{Recommend}()$
                \EndIf
            \ElsIf{$\textsf{op}_i = \textsc{SumSimplex}((\sigma_1, \sigma_2) \gets I)$}
                \If{$I \in \mathcal{I}_1$} \Comment{\textcolor{gray}{$I$ belongs to Player 1}}
                    \State $\xi[\sigma_1, \sigma_2] \gets \sum_{a\in A_I} \xi[(I, a), \sigma_2]$
                \Else
                    \State $\xi[\sigma_1, \sigma_2] \gets \sum_{a\in A_I} \xi[\sigma_1, (I, a)]$
                \EndIf
            \EndIf
        \EndFor
    \EndFunction
  \end{algorithmic}
  \begin{algorithmic}[1]
    \Function{ObserveLoss}{$\vec{\ell} \in \bbR^{|\Xi|}$}
        \State $\vec{\xi} \gets \vec{0}$
        \For{$i=n,\dots,1$}
            \If{$\textsf{op}_i = \textsc{FillSimplex}((\sigma_1, \sigma_2) \to I)$}
                \State $\vec{y} \gets \text{RM}_i.\textsc{Recommend}()$
                \If{$I \in \mathcal{I}_1$} \Comment{\textcolor{gray}{$I$ belongs to Player 1}}
                    \State $v \gets \sum_{a\in A_I} y[a]\cdot \ell[(I, a), \sigma_2]$
                    \State $\text{RM}_i.\textsc{ObserveLoss}((\vec{\ell}[(I, a), \sigma_2])_{a\in A_I})$
                \Else
                    \State $v \gets \sum_{a\in A_I} y[a]\cdot \ell[\sigma_1, (I, a)]$
                    \State $\text{RM}_i.\textsc{ObserveLoss}((\vec{\ell}[\sigma_1, (I, a)])_{a\in A_I})$
                \EndIf
                \State $\ell[\sigma_1, \sigma_2] \gets \ell[\sigma_1, \sigma_2] + v$
            \ElsIf{$\textsf{op}_i = \textsc{SumSimplex}((\sigma_1, \sigma_2) \gets I)$}
                \If{$I \in \mathcal{I}_1$} \Comment{\textcolor{gray}{$I$ belongs to Player 1}}
                    \For{$a\in A_I$}
                        \State $\ell[(I,a),\sigma_2] \gets \ell[(I,a),\sigma_2] + \ell[\sigma_1,\sigma_2]$
                    \EndFor
                \Else
                    \For{$a\in A_I$}
                        \State $\ell[\sigma_1,(I,a)] \gets \ell[\sigma_1, (I,a)] + \ell[\sigma_1,\sigma_2]$
                    \EndFor
                \EndIf
            \EndIf
        \EndFor
    \EndFunction
  \end{algorithmic}
\end{algorithm}

The pseudocode is given in \cref{algo:full rm}. \cref{line:update D 1,line:update D 6,line:update D 7} in \cref{algo:main} correspond to \textsc{FillSimplex} operations, while \cref{line:update D 4,line:update D 5} correspond to \textsc{SumSimplex} operations. With this notation, we can rewrite the decomposition of $\Xi$ in \cref{thm:decomposition} as a list of \textsc{FillSimplex} and \textsc{SumSimplex} operations:
\[
  \Xi = \textsf{op}_1\,.\,\textsf{op}_2\,.\,\cdots\,.\,\textsf{op}_n.
\]
In fact, our \cref{algo:full rm} operates on this representation. It simply corresponds to applying \cref{algo:rm} on the chin of operations, recursively. The algorithm is contingent on a choice of ``local'' regret minimizers $\text{RM}_i$ for each of the simplex domains $\Delta^{|A_I|}$ for each of the $\textsc{FillSimplex}(\sigma_1,\sigma_2)\to I)$ operations.

It's immediate to see that both \textsc{Recomment} and \textsc{ObserveLoss} run in linear time (in $\|\Xi\|$, that is, in the number of relevant sequence pairs in the game), provided that all ``local'' regret minimizers $\text{RM}_i$ run in linear time in the size of their respective simplexes. Furthermore, from \cref{prop:regret bound} we find that the regret of \cref{algo:full rm} is upper bounded by the sum of all regrets cumulated by the local regret minimizers $\text{RM}_i$. Putting these facts together, we conclude the following: 

\thmrmproperties*
    \section{Subgradient Descent Technique of~\citet{Farina19:Correlation}}\label{app:subgradient}

\cref{fig:subgradient infeasibility} shows the infeasibility of the iterates produced by the subgradient descent technique of~\citet{Farina19:Correlation}. It complements~\cref{fig:results}.

\begin{figure}[ht]
  \raisebox{1.8px}{\includegraphics[scale=.89]{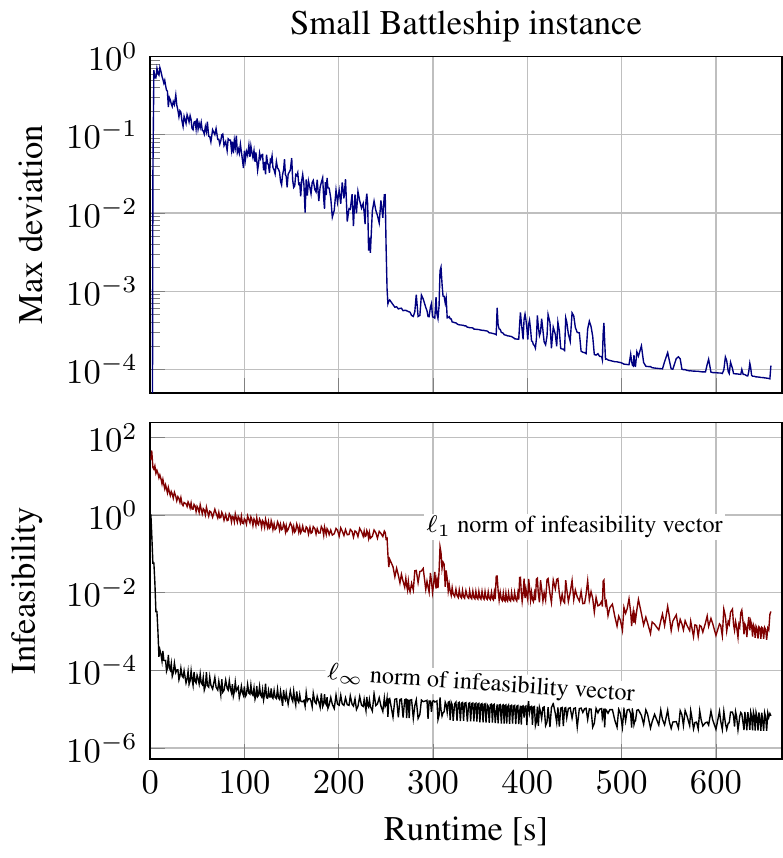}}
  \,
  \includegraphics[scale=.89]{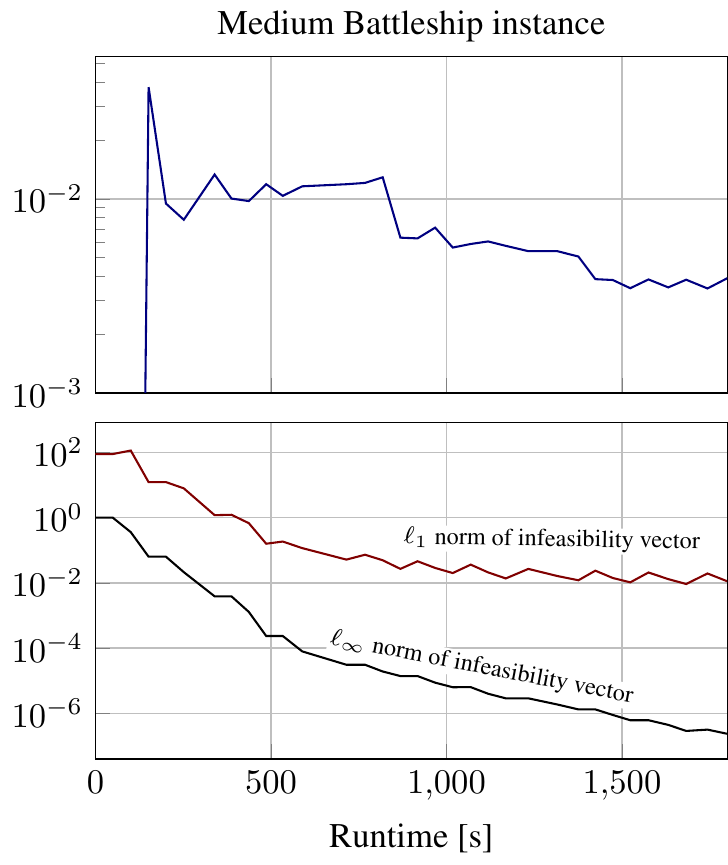}
  \caption{Infeasibility of the iterates produced by the subgradient descent technique of~\citet{Farina19:Correlation}. The infeasibility vector of an iterate $\vec{\xi}$ is defined as the vector of absolute differences between the left-hand and right-hand sides of all constraints that define $\Xi$ (\cref{def:xi}).}
  \label{fig:subgradient infeasibility}
\end{figure} 
\fi
\end{document}